\documentclass[draftclsnofoot,onecolumn,12pt]{IEEEtran}
\IEEEoverridecommandlockouts

\usepackage{kotex}

\usepackage{cite}
\usepackage{graphicx}
\usepackage{amsmath, amsfonts, amsthm}
\usepackage{xcolor}
\usepackage{booktabs, multirow}
\ifCLASSOPTIONcompsoc
    \usepackage[caption=false, font=normalsize, labelfont=sf, textfont=sf]{subfig}
\else
\usepackage[caption=false, font=footnotesize]{subfig}

\newtheorem{theorem}{Theorem}
\newtheorem{lemma}{Lemma}
\newtheorem{remark}{Remark}

\newcommand{\rarr}{\rightarrow}

\begin{document}

\title{Information-Theoretical Approach to Integrated Pulse-Doppler Radar and Communication Systems
}

\author{\IEEEauthorblockN{ Geon~Choi} and
 \IEEEauthorblockN{Namyoon~Lee},
 {\it Senior Member,~IEEE}

 \thanks{ G. Choi is with the Department of Electrical Engineering, POSTECH, South Korea, Pohang, (e-mail: {\texttt{simon03062@postech.ac.kr}}). N. Lee is with the School of Electrical Engineering, Korea University, South Korea, Seoul, (e-mail: {\texttt{namyoon@korea.ac.kr}})}}
 



\maketitle

\begin{abstract}

Integrated sensing and communication improves the design of systems by combining sensing and communication functions for increased efficiency, accuracy, and cost savings. The optimal integration requires understanding the trade-off between sensing and communication, but this can be difficult due to the lack of unified performance metrics. In this paper, an information-theoretical approach is used to design the system with a unified metric. A sensing rate is introduced to measure the amount of information obtained by a pulse-Doppler radar system.  An approximation and lower bound of the sensing rate is obtained in closed forms. Using both the derived sensing information and communication rates, the optimal bandwidth allocation strategy is found for maximizing the weighted sum of the spectral efficiency for sensing and communication. The simulation results confirm the validity of the approximation and the effectiveness of the proposed bandwidth allocation.
\end{abstract}

\section{Introduction}\label{sec1}

The persistent trend of spectrum scarcity in wireless communications has prompted the development of Integrated Sensing and Communication (ISAC), a revolutionary approach to provide high-speed communication and precise sensing services within limited spectrum availability. Unlike traditional communication and sensing systems, which are independently designed with separate hardware platforms and orthogonal spectrums, ISAC optimizes the system by sharing bandwidth and hardware \cite{zhang2021overview, cui2021integrating, liu2020joint, mishra2019toward, zheng2019radar}. This joint design approach enhances spectral efficiency and sensing accuracy, while also reducing hardware costs. ISAC has gained popularity as a solution for emerging applications, including autonomous driving in vehicular networks, indoor positioning with WiFi signals, and joint radar imaging and communication systems \cite{ ma2020joint, kumari2017performance, wang2016we, colone2013wifi, sankar2022beamforming, liu2022transmit}.

The integration of sensing and communication systems presents a major challenge due to the lack of understanding of the trade-off between their performance when sharing the same radio spectrum. While the classical Shannon theory \cite{shannon1948mathematical} outlines the limit of communication systems, it remains unclear what the maximum amount of reliable information can be obtained from unknown objects  (e.g., distances, velocities, angle of arrivals, and reflected power levels) in a sensing system for a given time-frequency resource. Historically, radar systems have been designed to accurately estimate the features of unknown objects, with mean-squared error (MSE) \cite{richards2010principles} and Cramer-Rao lower bound (CRLB)  \cite{poor1994introduction} being popular performance metrics. However, the lack of a connection between Fisher information and mutual information in information theory impedes the design of integrated sensing and communication systems with a unified performance metric. Thus, a unified performance metric is crucial for optimal integration of these heterogeneous systems.

In this paper, we present a new metric, the \textit{sensing rate}, for pulse-Doppler radar systems. The rate measures the amount of information about the target's environment obtained from the radar's returned signals. By using this rate, we derive the optimal bandwidth allocation strategy for ISAC systems that maximizes the combined spectral efficiency of both sensing and communication systems.


\subsection{Related Work}
There have been extensive efforts to evaluate information rates of a sensing system employing an information-theoretic approach. The pioneering work on incorporating the information-theoretical metric for sensing systems was found in \cite{chiriyath2015inner, paul2015extending, chiriyath2017radar}. The idea is to connect CRLB of the time delay estimation of a target to the differential entropy in measuring the target's uncertainty. The estimation radar rate is defined as the amount of the change for differential entropies between a returned signal and the signal after applying radar post-processing for the delay estimation. Armed with this tool, a trade-off between sensing and communication rates was characterized in different resource-sharing scenarios. In analogy to \cite{he2018performance}, the radar rate was derived using the CRLB of a target location in multi-antenna sensing systems, and the trade-offs between radar and communication rates were derived in a multi-antenna ISAC setting. Very recently, using the fact that CRLB constitutes a lower bound for the MSE of unbiased estimators \cite{xiong2022flowing, liu2021cramer}, the CLRB-rate region was derived to characterize the trade-off of ISAC systems.

An approach to understanding the information-theoretic limitation for an ISAC problem was initially introduced in \cite{kobayashi2018joint}, in which a transmitter sends a message to a receiver via a state-dependent stationary memoryless channel, and it estimates the state of the channel under the minimum distortion using a strictly causal channel output feedback. The trade-off between communication rates and state estimation accuracy in their setup was characterized using a capacity-distortion function. Recently, this approach has been extended to multi-user communication scenarios, including multiple access and broadcast channels in their subsequent works \cite{kobayashi2019joint, ahmadipour2021joint}. In addition, the simultaneous binary state estimation and data communication problem was considered in \cite{joudeh2022joint} and characterized a trade-off in terms of communication rate against the error exponent in detecting the binary channel state.

Although these prior works have opened a new direction toward optimally integrating the sensing and communication systems from an information-theoretical viewpoint, these approaches still have limitations to understanding the sensing rate for widely-used pulse-Doppler radar systems (e.g., radar imaging applications). A pulse-Doppler radar is a sensing system that simultaneously determines a target's range, velocity, and reflected power by sending multiple pulses repeatedly during a coherent processing interval (CPI) \cite{richards2010principles, mishra2019sub, herman2009high, bajwa2011identification}. In a pulse-Doppler radar, the range and doppler resolutions are determined as a function of signal bandwidth and coherent processing interval (CPI), respectively. Instead of CRLB for the parameter estimation, these resolutions are more decisive because they limit the accuracy of the range and doppler estimation. Consequently, these parameters should be incorporated when defining the sensing rate. Unfortunately, all prior works in \cite{chiriyath2015inner, paul2015extending, chiriyath2017radar, he2018performance, xiong2022flowing, liu2021cramer, kobayashi2018joint, kobayashi2019joint, ahmadipour2021joint, joudeh2022joint} do not consider these resolution effects when defining the sensing rates.

In a sensing system, a reflected signal by a target, referred to as a target signal, can be interpreted as \textit{a transmitted signal by a target in environments}; thereby, it is essential to exploit the distribution of the target signal to measure the uncertainty of a target in a sensing system. The Swerling target models have been popularly used in radar systems, which combines a probability density function (PDF) and a decorrelation time for the target radar cross-section (RCS) \cite{swerling1960probability}. For instance, In Swerling I and Swerling II target models, the total RCS is assumed to be the sum of many independent small scatterers, each with approximately equal individual RCS power in analogy to modeling a rich-scattering channel in communications. Furthermore, in Swerling I, this total RCS  is assumed to be constant during CPI, and the amplitude is distributed as Rayleigh in analogy to a block fading model in communications. In Swerling II, however, the total RCS is assumed to change with every pulse per scan, i.e., a fast-fading channel model in communications. As a result, unlike the previously defined sensing rates in \cite{chiriyath2015inner, paul2015extending, chiriyath2017radar, he2018performance, xiong2022flowing, liu2021cramer}, the sensing rate also requires incorporating the statistical distributions of the target signals. 

\subsection{Contributions}

How much information rate is acquired by the pulse-Doppler radar system from targets in an environment? In this paper, we address this question by introducing a new concept - the sensing rate of a range-Doppler radar system. This rate is calculated as a combination of range, Doppler resolutions, and target's statistical distributions. Unlike previous works in \cite{chiriyath2015inner, paul2015extending, chiriyath2017radar, he2018performance, xiong2022flowing, liu2021cramer} that utilize CRLB, our approach is based on more practical models and assumptions of the pulse-Doppler radar system. The key contributions of this paper are outlined below:

\begin{itemize}

\item We first introduce the concept of sensing rate, which quantifies the amount of information obtained from a target in terms of bits per channel use. We then derive the MSE-optimal target detector under the Bernoulli-Gaussian target distribution, demonstrating that it generalizes the traditional Neyman-Pearson detector. By utilizing the I-MMSE relationship from \cite{guo2005mutual}, we derive an exact expression for the sensing rate in integral form. This expression reveals that the sensing rate is comprised of two components, related to the target's location and fluctuation information. By leveraging this insight, we establish a connection between information-theoretic measures and estimation-theoretic measures, such as the minimum target detection error probability and the minimum MSE of target RCS.


\item We also derive a tight approximation and a lower bound of the sensing rate in closed-form expressions that maintain full generality. These expressions are valuable in comprehending how the sensing rate varies based on system parameters. The derived lower bound of the sensing rate clearly indicates that the rate is directly proportional to the target uncertainty in the range-Doppler plane and the target signal-to-noise ratio (SNR). Our simulations validate the accuracy of the derived analytical expressions.

\item To shed light on the derived sensing rate, we present an optimal bandwidth allocation strategy for an ISAC system. To the best of our knowledge, no prior studies have explored the optimal bandwidth allocation problem for an ISAC system that maximizes the weighted sum spectral efficiency, due to the lack of a definition of sensing rates in terms of bits per channel use. With the derived information sensing rate and Shannon's rate, our optimal bandwidth allocation strategy seeks to maximize the weighted sum spectral efficiency of both the sensing and communication rates. By deriving the first-order optimality condition for the weighted sum spectral efficiency, our strategy suggests that the ISAC system should allocate more bandwidth to the sensing system as the target uncertainty and signal power increase. This confirms our intuition, and simulation results show that our derived bandwidth allocation solution outperforms traditional uniform bandwidth allocation strategies.

\end{itemize}

\section{System Model}
In this section, we present an ISAC system model. As illustrated in Fig.~\ref{fig:ISAC system model}, an ISAC transceiver sends a pulse train and receives echoes for sensing an unknown environment. Meanwhile, it also transmits a message to a receiver for communications. The total system bandwidth is assumed to be $B$, and the sensing and communication systems use a fraction of the total bandwidth, i.e., $B_{\rm  s}=\lambda B$  and  $B_{\rm c}=(1-\lambda) B$ for $\lambda \in [0,1]$. We explain the signal models for the sensing system, and then define the performance metric for the ISAC system.

\begin{figure}[t]
\centering
\includegraphics[width=0.8\columnwidth]{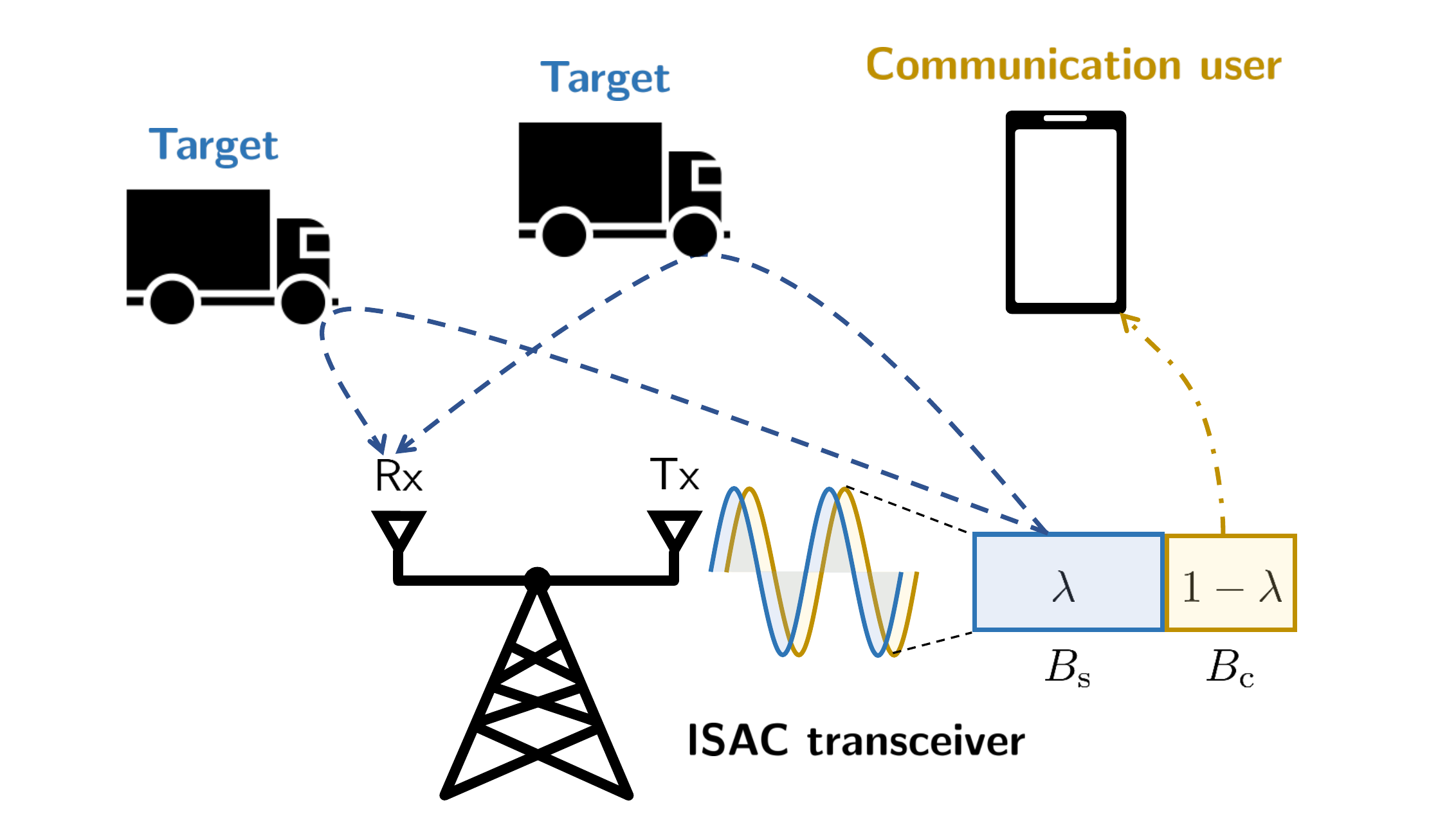}
\caption{An illustration of a ISAC system model.}
\label{fig:ISAC system model}
\end{figure}

\subsection{Pulse-Doppler Radar System Model}

We first explain the continuous-time baseband signal models and the relevant assumptions for pulse-Doppler radar systems. Then, we present the equivalent discrete-time signal models for pulse-Doppler radar systems when applying  conventional radar post-processing.     

\subsubsection{Continuous-time signal model}
The radar sends signal $x(t)$ comprised of $M$ narrowband pulses $g(t)$ with pulse repetition interval (PRI) $T_{\rm p}$ during coherent processing interval (CPI) as
\begin{align}
    x(t) = \sum_{m=0}^{M-1}g(t-mT_{\rm p}), \qquad 0\le t\le MT_{\rm p},
\end{align}
where the pulse $g(t)$ has the bandwidth of $B_{\rm s}={1} / {T_{\rm s}}$ and transmit energy $\int_{0}^{T_{\rm s}} |g(t)|^2{\rm d}t=P_{\rm t}T_{\rm s}$. This transmitted signal is reflected by targets in an unknown environment.

We adopt the Swerling I target model \cite{swerling1960probability}. This target model has been widely used to capture a static environmental scenario in the time scale of CPI. In this model, $L$ non-fluctuating point targets are well-spread during CPI, and $L$ targets are assumed to be statistically independent. Each target is defined by three parameters: i) range $d_{\ell}$, ii) radial velocity $v_{\ell}$, and iii) complex amplitude $H_{\ell}$ for $\ell\in [L]$. The range and the radial velocity of target $\ell$ correspond to round-trip time-delay $\tau_{\ell}={2d_{\ell}} / {c}$ and normalized doppler frequency $\nu_{\ell}={2 f_c v_{\ell}} / {c}$. The complex amplitude $H_{\ell}$ is assumed to be drawn from IID complex Gaussian with mean ${\mu}_{T,\ell}$ and variance ${\sigma}_{T,\ell}^2$, i.e., $H_{\ell}\sim \mathcal{CN}({ \mu}_{T,\ell}, { \sigma}_{T,\ell}^2)$ for $\ell\in [L]$.

We assume that the radar system operates under the following target assumptions: 
\begin{itemize}
\item  Assumption 1: The target is static during CPI; this condition holds when the target delays $\{\tau_{\ell}\}_{\ell=1}^L$, Doppler frequencies $\{\nu_{\ell}\}_{\ell=1}^L$, and the complex amplitude $\{H_{\ell}\}_{\ell=1}^L$ are constant during CPI. In other words, the target state changes slowly in time and frequency.
\item  Assumption 2: The $L$ delay-Doppler pairs $\{\tau_{\ell}, \nu_{\ell}\}_{\ell=1}^L$ are uniquely localized in the delay-Doppler plane defined by $[0,T_{\rm p}]\times \left[ -\frac{1}{2T_p}, \frac{1}{2T_p}\right]$. This fact implies that any two distinct targets cannot be overlapped in the plane, i.e., no ambiguity in the target presence.
\end{itemize}


Under the aforementioned assumptions, the continuous-time baseband received signal from the target reflections during CPI is given by
\begin{align}
	y(t) &= \sum_{m=0}^{M-1}\sum_{\ell=1}^{L}\sqrt{\beta_{\ell}}H_{\ell} g(t-\tau_{\ell}-mT_{\rm p})  e^{j2\pi m\nu_{\ell} t} + w(t), \label{eq:rx_signal}
\end{align}
where $w(t)$ is the complex additive white Gaussian noise with mean zero and variance $B_{\rm s}N_0$, which is defined as the product of signal bandwidth $B_{\rm s}$ and noise spectral density $N_{0}$. In addition, $\beta_{\ell}$ is the received signal power from the $\ell$th target reflection, which is defined by the classical radar equation:
\begin{align}
	\beta_{\ell} = \frac{P_{\rm t}G_{\rm t}}{4\pi d_{\ell}^2}\frac{\sigma_{\rm rcs} A_{\rm eff}}{4\pi d_{\ell}^2}, \label{eq:pathloss}
\end{align}
where $d_{\ell}= {\tau_{\ell} c} / {2}$ is the distance from the $\ell$th target to the receiver, $\sigma_{\rm rcs}$ denotes the RCS of a target measured in $m^2$, and $A_{\rm eff}$ is the effective area of the receiving antenna. 

Let $x_m(t)$ be the noise-free received signal during the $m$th PRI, which is defined as
 \begin{align}
	x_m(t) =\sum_{\ell=1}^{L}\sqrt{\beta_{\ell}}H_{\ell} g(t-\tau_{\ell}-mT_{\rm p})e^{-j2\pi m  \nu_{\ell} T_{\rm p} }  \label{eq:target_signal}
\end{align}  
for $mT_{\rm p} \leq  t < (m+1)T_{\rm p}$. This noise-free returned signal is of significance because it can be interpreted as the transmitted signal by the environment, i.e., a target signal. As a result, from a communication system perspective, the target signal plays the role of the transmitted signal by the targets. Using $x_m(t)$, we can rewrite the received signal in \eqref{eq:rx_signal} during the $m$th PRI as
\begin{align}
	y_m(t) = x_m(t) + w_m(t),  \label{eq:bbrx}
\end{align}
where $x_m(t)$ and $w_m(t)$ are the fraction of $x(t)$ and $w(t)$ for $mT_{\rm p} \leq  t < (m+1)T_{\rm p}$. We shall focus on this continuous-time baseband received signal model in the sequel.

\subsubsection{Radar post-processing}

The main task of the pulse-Doppler radar is to identify $3L$ parameters $\{\tau_{\ell},\nu_{\ell}, H_{\ell}\}_{\ell=1}^L$ from the received signal $y(t)$ during CPI as in \eqref{eq:rx_signal}.  To this end, the classical pulse-Doppler radar system performs two sequential post-processing: i) fast-time and ii) slow-time processing. Applying the post-processing, we can obtain a discrete-time equivalent signal model.
  
{\bf Fast-time processing:} 
The fast-time processing allows obtaining the discrete-time samples of the returned signal using matched filtering and sampling. After the matched filtering and sampling at the Nyquist rate $T_{\rm s}= {1} / {B_{\rm s}}$, the receiver obtains the discrete-time samples. 

Let ${\bar Y}_{n,m}$, ${\bar X}_{n,m}$, and ${\bar W}_{n,m}$ be the discrete-time received, target, and noise signals for the $n$th range bin during the $m$th PRI. Using these notations, we can rewrite the continuous-time input-output relationship in \eqref{eq:bbrx} as
\begin{align}
	{\bar Y}_{n,m} = {\bar X}_{n,m} + {\bar W}_{n,m},
\end{align} 
where ${\bar W}_{n,m}$ follows the complex Gaussian with zero-mean and variance $B_{\rm s}N_0$, i.e., ${\bar W}_{n,m}\sim\mathcal{CN}(0, B_{\rm s}N_0)$ and ${\bar X}_{n,m}$ contains the target information for the relevant parameters $\{\tau_{\ell},\nu_{\ell}, H_{\ell}\}_{\ell=1}^{L}$ as defined in \eqref{eq:target_signal}. Suppose the $\ell$th target exists in the $n$th range bin, i.e., $\left\lfloor {\tau_{\ell}} / {T_{\rm s}} \right \rfloor= n$ for $n\in \{0,\ldots,N-1\}$.  In this case,  the target signal is defined as ${\bar X}_{n,m}= 	\sqrt{\beta_{\ell}} H_{\ell}e^{-j 2\pi m \nu_{\ell} T_{\rm p}}$. Whereas, if a target is absent in the $n$th range bin, ${\bar X}_{n,m}$ is assumed to be null, i.e., ${\bar X}_{n,m}=0$.  Using unit impulse function $\delta[n]$, we can define the discrete-time target signal ${\bar X}_{n,m}$ in compact form:
\begin{align}
{\bar X}_{n,m}  =\sum_{\ell=1}^L	\sqrt{\beta_{\ell}} H_{\ell} \delta [n-U_{\ell}]   e^{-j 2\pi m \nu_{\ell} T_{\rm p}},  \label{eq:target_discrete}
\end{align}
where $U_{\ell}=\left\lfloor {\tau_{\ell}} / {T_{\rm s}} \right \rfloor \in \{0,\ldots, N-1\}$ is discrete random variable indicating the range information for the $\ell$th target.

{\bf Slow-time processing:} In the pulse-Doppler radar system, the receiver also performs the matched filtering along the pulse dimension to increase the received power of the target signals.  Recall that $0 \leq \nu_{\ell} T_{\rm p} \leq 1 $ and the Doppler resolution is ${1} / {MT_{\rm p}}$. The limited Doppler resolution allows distinguishing target Doppler frequencies in the discrete set, i.e., $\nu_{\ell} T_{\rm p} \in \left\{0,\frac{1}{M}, \ldots, \frac{M-1}{M}\right\}$.  With the discrete Doppler frequencies, the target signal in \eqref{eq:target_discrete} boils down to
	\begin{align}
		{\bar X}_{n,m}  =\sum_{\ell=1}^L\sqrt{\beta_{\ell}}	H_{\ell} \delta [n-U_{\ell}]  e^{-j 2\pi  m \frac{V_{\ell}}{M}},
	\end{align}
where $V_{\ell}\in \{0,\ldots,M-1\}$ represents the Doppler bin index of the $\ell$th target. For given range bin index $n\in \{0,1,\ldots, N-1\}$, the receiver takes an $M$-point discrete Fourier transform (DFT) operation as
 \begin{align}
	A_{n,k} &= \frac{1}{\sqrt{M}} \sum_{m=0}^{M-1}{\bar X}_{n,m}e^{-j \frac{2\pi k}{M}m} \nonumber\\
	&=\sum_{\ell=1}^L\ \sqrt{\frac{\beta_{\ell}}{M}} 	H_{\ell} \delta [n-U_{\ell}] \sum_{m=0}^{M-1} e^{-j m\left(  \frac{ 2\pi V_{\ell}}{M} - \frac{2\pi k}{M} \right)}.
\end{align}
Using the fact that $ \sum_{m=0}^{M-1} e^{-j m\left(  \frac{ 2\pi V_{\ell}}{M} - \frac{2\pi k}{M} \right)} =M\delta [ k- V_{\ell}  ]$, we can represent  the $(n,k)$ entry of the target signal as
\begin{align}
	A_{n,k} =\sum_{\ell=1}^L \sqrt{M \beta_{\ell}}  	H_{\ell} \delta [ n-U_{\ell} ] \delta [ k- V_{\ell}  ]. \label{eq:discrete_target_model}
\end{align}

\subsubsection{Discrete-time signal model}  Let ${\sf snr}_n$ be the signal-to-noise ratio (SNR) of the $n$th range bin. From the radar equation in \eqref{eq:pathloss} and the target signal model in \eqref{eq:discrete_target_model}, the SNR under the target's presence in the $n$th range bin is defined as
\begin{align}\label{eqn:radar discrete snr}
	{\sf snr}_n	&= \frac{M}{ B_{\rm s} N_o}\frac{P_{\rm t}G_{\rm t}\sigma_{\rm rcs} A_{\rm eff}}{\left(4\pi\right)^2 {\bar d}_{n}^4},
\end{align}
where ${\bar d}_{n}$ denotes the average distance from the $n$th range bin to the receiver, i.e., ${\bar d}_{n}= \frac{(2n+1)T_{\rm s} c}{4}$.   Denoting $Y_{n,k}=\frac{1}{\sqrt{M}} \sum_{m=0}^{
M-1} {\bar Y}_{n,m}e^{j \frac{ 2\pi mk}{M}  }$ and $Z_{n,k}=  \frac{1}{\sqrt{M  }}\sum_{m=0}^{M-1}{\bar W}_{n,m}e^{j \frac{ 2\pi mk}{M}  }$, the discrete-time input-output signal relationship for the $(n,k)$ range-Doppler bin is given by
\begin{align}\label{eqn:discrete_observation_model}
	Y_{n,k} = \sqrt{{\sf snr}_n}X_{n,k} + Z_{n,k}, 
\end{align} 
where $ Z_{n,k}$ is the normalized noise, which follows the complex Gaussian with zero-mean and unit variance, i.e., $ Z_{n,k}\sim \mathcal{CN}(0,1)$. From \eqref{eq:discrete_target_model}, we model that target returned signal $X_{n,k}$ is distributed as a Bernoulli-Gaussian (BG) random variable, i.e., 
\begin{align}
	f_{X_{n,k}}(x_{n,k}) = (1-\gamma) \delta(x_{n,k}) +\gamma \frac{1}{ \pi \sigma_{\sf T}^2}e^{-\tfrac{(x_{n,k}-\mu_{\sf T})^2}{\sigma_{\sf T}^2}}, \label{eqn:target signal pdf}
\end{align}
where $\gamma \in (0,1)$ determines the probability that $X_{n,k}$ is a non-zero value. We shall use a short notation $X_{n,k}\sim {\sf BG}(x|\mu_{\sf T},\sigma^2_{\sf T},\gamma)$ in the sequel.

\begin{remark} \rm
    The parameter $\gamma$ determines the sparsity level of $X_{n,k}$ in the range-Doppler plane. This parameter $\gamma$ contains a prior knowledge about environment. For instance, if the system is blind to the environment, one can choose $\gamma = 1/2$. However, when $L$ targets on average have been detected in the previous CPIs, one can choose $\gamma =\frac{L}{MN}$. In this work, we focus on the specific CPI and assume $\gamma$ to be a pre-determined parameter.
\end{remark}




\subsection{Performance Metrics}

We shall characterize the {\it sensing spectral efficiency} of the pulse-Doppler radar system. An information theoretical approach to measure the information rate between $X_{n,k}$ and $Y_{n,k}$ is given by
\begin{align}
 R_{n,k}({\sf snr}_n)=I(X_{n,k}; Y_{n,k}  )~~({\rm bits/sec/Hz} ),
\end{align}
for $n\in \{0,1,\ldots, N-1\}$ and $k\in \{0,1,\ldots, M-1\}$. It is noteworthy that the number of range bins $N$ depends on the bandwidth ratio $\lambda$ according to the Rayleigh criterion, which states that the range resolution is the inverse of the signal bandwidth \cite{richards2010principles}. This mutual information measures the amount of information acquired from the environment placed at the $(n,k)$ range-Doppler bin. Since the mutual information differs across the range bins only, we define the average sensing rate as  \begin{align}\label{eqn:sensing rate}
 R_{\rm s}(B_{\rm s}, {\sf snr}_n)=\frac{B_{\rm s}}{NM} \sum_{n=0}^{N-1}\sum_{k=0}^{M-1}R_{n,k}(	{\sf snr}_n) ~~({\rm bits/sec}).
\end{align}
This average information rate is the amount of information attained by sensing the target environment per second.

For the communication system, we consider the Gaussian channel. In the Gaussian channel, the communication rate with bandwidth $B_{\rm c}$ is given by
\begin{align}
	R_{\rm c}(B_{\rm c}, {\sf snr_{c}}) &= B_{\rm c} \log_2\left(1+{\sf snr}_{c} \right) ~~({\rm bits/sec} ), 
\end{align}
where ${\sf snr}_{\rm c}$ is the received SNR of communication system.
Intuitively, increasing bandwidth for the range-Doppler radar system can improve sensing rate $R_{\rm s}$ while decreasing communication rate $R_{\rm c}$. As a result, it is crucial to understand a trade-off between the sensing and communication rates according to the bandwidth ratio $\lambda$. To measure this trade-off, we define a weighted sum spectral efficiency as \begin{align}
	R_{\rm sum} = \frac{w_{\rm s}R_{\rm s}(B_{\rm s}, {\sf snr}_n)+ w_{\rm c}R_{\rm c}(B_{\rm c}, {\sf snr}_c)}{B_{\rm s} + B_{\rm c}} ~~({\rm bits/sec/Hz}),
	\label{eqn:weighted sum rate}
\end{align}
where $w_{\rm s}$ and $w_{\rm c}$ are weights for sensing and communication rates, respectively. 

\begin{remark} \rm
    The system weights $w_{\rm s}$ and $w_{\rm c}$ can be chosen according to priority aspects of sensing and communication functions. For instance, the communication system for emergency broadcasting system is assigned higher than system for usual data transmission. 
The intrinsic disparity of two system can be incorporated into those weights \cite{chiriyath2017radar}. The sensing system requires more power to obtain 1 bit than communication system because sensing system is regarded as uncooperative communication system. The high power demanding system weight can be assigned to be higher than other systems.  
 
\end{remark}

\section{Bayesian Target Estimation and Detection }
In this section, we first derive the MSE-optimal target estimation technique. Then, we show that both a maximum a posterior (MAP) target detector, which minimizes the average target detection error, and Neyman-Pearson detector are  special cases of the MSE-optimal target estimator.  

%



\subsection{MMSE Target Estimation}
The main task of the range-Doppler radar is to identify the non-zero entries from the measurement matrix ${\bf Y}\in \mathbb{C}^{N\times K}$, which contains the target information during CPI. Since target signals $X_{n,k}$ and  $X_{n',k'}$ are assumed to be statistically independent of $n\neq n'$ and $k\neq k'$, it is sufficient to design a target estimator for individual range-Doppler bin. 

Recall the received signal of the $(n,k)$ range-Doppler bin $Y_{n,k}=X_{n,k}+Z_{n,k}$ in \eqref{eqn:discrete_observation_model}. The mean and variance of $X_{n,k}\sim {\sf BG}(x|\mu_{\sf T},\sigma_{\sf T}^2,\gamma)$ are given by
\begin{align}
\mathbb{E}[X_{n,k}] =\gamma \mu_{\sf T}~~~~{\rm and}~~~~
{\sf Var}[X_{n,k}] =\gamma(1-\gamma)\vert\mu_{\sf T}\vert^2 +  \gamma \sigma_{\sf T}^2. \label{eq:mean_var}
\end{align}
Under the Gaussian measurement noise $Z_{n,k}\sim \mathcal{N}(0,1)$, it is evident that the posterior distribution $P_{X_{n,k}|Y_{n,k}}(x_{n,k}|y_{n,k})$ also follows a BG distribution with different parameters $({\hat \mu}_n, {\hat \sigma}_n^2, {\hat \gamma})$ because $Y_{n,k}$ is distributed as a BG random variable. The parameter ${\hat \mu}_n$ is the posterior mean conditioned that $X_{n,k}\neq 0$, which is computed as
\begin{align}
    {\hat \mu}_n&=\mathbb{E}\left[X_{n,k}|Y_{n,k}=y_{n,k}, X_{n,k}\neq 0  ; {\sf snr}_n\right] \nonumber\\
    &=\mu_{\sf T}+\frac{\sqrt{{\sf snr}_n}\sigma_{\sf T}}{1+{\sf snr}_n\sigma_{\sf T}^2 }\left(y_{n,k}-\sqrt{{\sf snr}_n} \mu_{\sf T} \right). \label{eq:cond_mean}
\end{align}
The parameter ${\hat \sigma}_n^2$ is the posterior variance conditioned that $X_{n,k}\neq 0$. This conditional variance is equivalent to the well-known MMSE under the Gaussian prior, i.e., 
\begin{align}
    {\hat \sigma}_n^2&={\sf Var}\left[X_{n,k}|Y_{n,k}, X_{n,k}\neq 0  ; {\sf snr}_n \right] \nonumber\\
    &=\frac{\sigma_{\sf T}^2}{1+{\sf snr}_n\sigma_{\sf T}^2}.\label{eq:cond_var}
\end{align}
The last parameter ${\hat \gamma}_n$ is the posterior probability of a target existence, which is given by the following lemma. 

\begin{lemma}\label{lem1} Given $Y_{n,k}=y_{n,k}$, the posterior probability that a target signal is non-zero is given by
\begin{align}
		{\hat \gamma}_n (y_{n,k})&= \mathbb{P}[X_{n,k}\neq 0 |Y_{n,k}=y_{n,k}] \nonumber\\
&= \left[1+\frac{1-\gamma}{\gamma}({\sf snr}_n\sigma_{\sf T}^2+1)e^{-\vert y_{n,k}\vert^2+\frac{\vert y_{n,k}-\sqrt{{\sf snr}_n}\mu_{\sf T}\vert^2}{{\sf snr}_n\sigma_{\sf  T}^2+1}}\right]^{-1}. \label{eq:lem1}
\end{align}
\end{lemma}
 
\begin{proof}
 See Appendix~\ref{appendix:proof of mmse estimator}.
\end{proof}

 
The following theorem presents the MSE-optimal target estimator and the corresponding MMSE. 
 
\begin{theorem}\label{thm:mmse estimator}
The MMSE estimator for the target signal in the $(n,k)$ range-Doppler bin is given by
\begin{align}\label{eqn:mmse estimator}
    {\hat X}_{n,k}^{\sf mmse}(y_{n,k} ; {\sf snr}_n)
    &=\mathbb{E}\left[X_{n,k}|Y_{n,k}=y_{n,k}  ; {\sf snr}_n\right] \nonumber\\
    &= \frac{ \mu_{\sf T}+\frac{\sqrt{{\sf snr}_n}\sigma_{\sf T}}{1+{\sf snr}_n\sigma_{\sf T}^2 }\left(y_{n,k}-\sqrt{{\sf snr}_n} \mu_{\sf T} \right) }{1+\frac{1-\gamma}{\gamma}({\sf snr}_n\sigma_{\sf T}^2+1)e^{-\vert y_{n,k}\vert^2+\frac{\vert y_{n,k}-\sqrt{{\sf snr}_n}\mu_{\sf T}\vert^2}{{\sf snr}_n\sigma_{\sf  T}^2+1}}}.
\end{align}
The resultant MMSE is given by
\begin{align}
    {\sf mmse}_{n,k} ({\sf snr}_n)
    &= \mathbb{E}\left[\left( X_{n,k}-  {\hat X}_{n,k}^{\sf mmse}(Y_{n,k})\right)^2  ; {\sf snr}_n\right] \nonumber\\
    &= \int_{y\in \mathbb{C}} \left\{ {\hat \gamma}_n (y)(1-{\hat \gamma}_n (y))\left\vert\mu_{\sf T}+\frac{\sqrt{{\sf snr}_n}\sigma_{\sf T}}{1+{\sf snr}_n\sigma_{\sf T}^2 }\left(y-\sqrt{{\sf snr}_n} \mu_{\sf T} \right)\right\vert^2  \right. \nonumber\\
    &~~~ + \left. {\hat \gamma}_n (y)\left[\frac{\sigma_{\sf T}^2}{1+{\sf snr}_n\sigma_{\sf T}^2}\right]^2 \right\} f_{Y_{n,k}}(y)  {\rm d}y, \label{eq:mmse}
\end{align}
where  $f_{Y_{n,k}}(y) $ is a probability density function of a mixture-Gaussian random variable $Y_{n,k}$, i.e., \begin{align}
	f_{Y_{n,k}}(y) = (1-\gamma)  \frac{1}{ \pi  }e^{-|y|^2} +\gamma \frac{1}{ \pi \left(  {\sf snr}_n\sigma_{\sf T}^2+1\right)}e^{-\frac{\vert y-\sqrt{{\sf snr}_n}\mu_{\sf T}\vert^2}{ {\sf snr}_n\sigma_{\sf T}^2+1}}. \label{eqn:received signal pdf}
\end{align}

\end{theorem}
 
\begin{proof}
The proof is direct by invoking \eqref{eq:cond_mean}, \eqref{eq:cond_var}, and \eqref{eq:lem1} into \eqref{eq:mean_var}.   \end{proof}

The derived MMSE estimator in Theorem 1 guarantees the optimal reconstruction quality in terms of MSE if the target signal distribution follows a BG distribution,  $X_{n,k}\sim {\sf BG}(x|\mu_{\sf T},\sigma_{\sf T}^2,\gamma)$. The noticeable difference with the classical MMSE estimator for the Gaussian prior is a multiplicative parameter by ${\hat \gamma_n}$ defined in \eqref{eq:lem1}, which captures the chance to detect a target for given measurement $y_{m,k}$. The derived MMSE estimator can also be interpreted with a lens through a denoising function in imaging processing. The measurement data matrix ${\bf Y} \in \mathbb{C}^{N\times M}$ can be regarded as containing a  two-dimensional complex target image ${\bf X} \in \mathbb{C}^{N\times M}$. The leading role of the MMSE estimator is to reduce noise signals; thereby, it constructs a more clear target image than ${\bf Y}$. This role differs from classical radar detectors, whose primary goal is to detect the target presence in a range-Doppler plane.

\subsection{Connection to Classical Radar Detectors}

To provide a more clear understanding of the derived MMSE radar detector, it is instructive to compare it with two classical radar detectors: i) MAP and ii) Neyman-Pearson (NP) detectors. Let us consider a binary target detector $g(\cdot):\mathbb{C}\rightarrow \{0,1\}$ which maps received signal $Y_{n,k}$ to a binary value according to the target presence. Armed with this binary detector, we define two types of probability errors: i) false alarm probability $P_{\sf FA}$ and misdetection probability $P_{\sf MD}$ as
\begin{align}
    P_{\sf FA}({\sf snr}_{n}) &= \mathbb{P}\left[g(Y_{n,k}) = 1 \mid X_{n,k} = 0 \right] \text{ and } P_{\sf MD}({\sf snr}_{n})  = \mathbb{P}\left[g(Y_{n,k})= 0 \mid X_{n,k} \neq 0 \right].
\end{align}
Then, an average probability error is given by
\begin{align}
P_{\sf E}(\gamma, {\sf snr}_{n })  = \gamma P_{\sf MD}({\sf snr}_{n })  + (1-\gamma) P_{\sf FA}({\sf snr}_{n}).
\end{align}

The MAP detector that minimizes the average probability error, $g_{\sf MAP}(\cdot) = \arg\min_{g(\cdot)} P_{\sf E}$, is given by 
\begin{align}
    g_{\sf MAP}(Y_{n,k}) &= \begin{cases}
    0, & \frac{\mathbb{P}(X_{n,k} \neq  0 \vert Y_{n,k})}{\mathbb{P}(X_{n,k} = 0 \vert Y_{n,k})}   < 1, \\
    1, & \frac{\mathbb{P}(X_{n,k} \neq  0 \vert Y_{n,k})}{\mathbb{P}(X_{n,k} = 0 \vert Y_{n,k})}   > 1.
    \end{cases} \label{eqn:MAP detector}
\end{align}
From Lemma 1, the ratio of the posterior probabilities can be expressed as $\frac{\mathbb{P}(X_{n,k} \neq  0 \vert Y_{n,k}=y_{n,k})}{\mathbb{P}(X_{n,k} = 0 \vert Y_{n,k}=y_{n,k})} =\frac{{\hat \gamma}_n (y_{n,k})}{1-{\hat \gamma}_n (y_{n,k})}$. Using this ratio, the MAP decision rule boils down to
\begin{align}
    g_{\sf MAP}(Y_{n,k}) &= \begin{cases}
    0, & \frac{{\hat \gamma}_n (Y_{n,k})}{1-{\hat \gamma}_n (Y_{n,k})} < 1, \\
    1, & \frac{{\hat \gamma}_n (Y_{n,k})}{1-{\hat \gamma}_n (Y_{n,k})} > 1.
    \end{cases} \label{eqn:MAP detector 2}
\end{align}
When establishing the MAP detector,  the posterior probability of the target presence ${\hat \gamma}_n (y_{n,k})$ in \eqref{eq:lem1} is sufficient information. The MMSE estimator, however, requires additional information about ${\hat \mu}_n$ in \eqref{eq:cond_mean} and ${\hat \sigma}_n^2$ in \eqref{eq:cond_var} beyond ${\hat \gamma}_n (y_{n,k})$. This fact reveals that the classical MAP detector is a special case of the MMSE estimator when discarding information about ${\hat \mu}_n$  and ${\hat \sigma}_n^2$.

%

\begin{figure}[t]
\centering
\includegraphics[width=0.8\columnwidth]{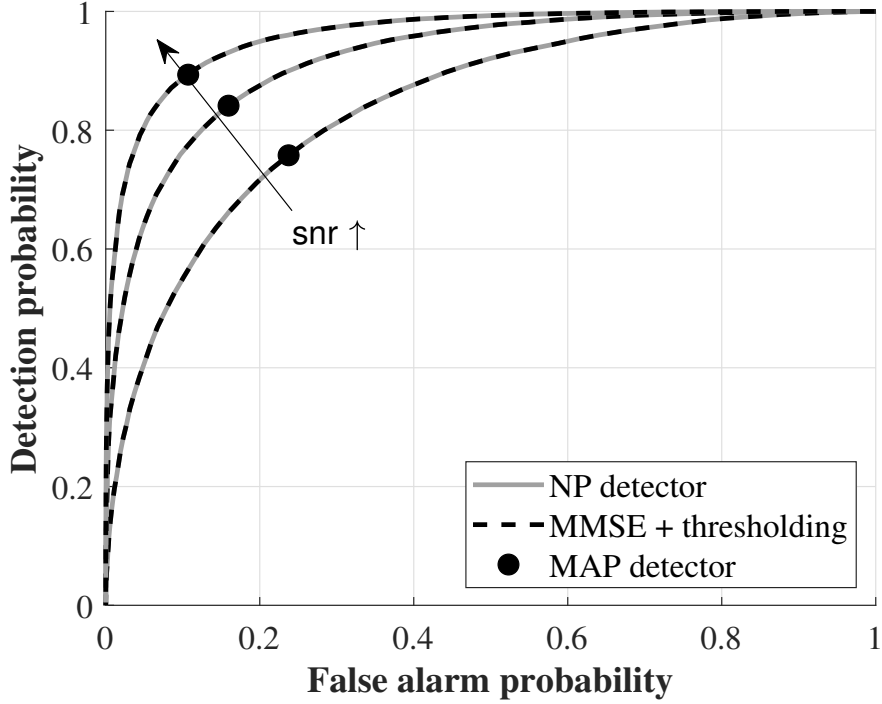}
\caption{ROC curves for three different types of detectors: i) NP, ii) MMSE estimator with hard threshold, and iii) MAP detectors. Target parameters are set to be $\gamma = 0.5$, $\mu_{\sf T} = 1$, $\sigma_{\sf T}^2 = 10^{-1}$, ${\sf snr}_n \in \{0, 3, 5\}$ [dB].}
\label{fig:ROC}
\end{figure}

Unlike the Bayesian approach using a prior distribution, a classical detector in radar theory does not exploit prior information of the target signal because of the absence of such knowledge in general.  According to the celebrated NP rule, the optimal binary detector $g_{\sf NP}(Y_{n,k})$ with threshold $\delta$ can be defined with the likelihood ratios as 
\begin{align}
    g_{\sf NP}(Y_{n,k}) &= \begin{cases}
    0, & \frac{\mathbb{P}(Y_{n,k} \vert X_{n,k} \neq 0)}{\mathbb{P}(Y_{n,k} \vert X_{n,k} = 0)} < \delta, \\
    1, & \frac{\mathbb{P}(Y_{n,k} \vert X_{n,k} \neq 0)}{\mathbb{P}(Y_{n,k} \vert X_{n,k} = 0)} > \delta.
    \end{cases} \label{eqn:NP detector}
\end{align}
The threshold parameter $\delta$ can be chosen to minimize $P_{\sf FA}$ subject to a constraint of $P_{\sf MD}=\alpha$ with $\alpha>0$, which is referred to as NP detector. This NP detector can be deduced from the posterior probability of the target presence ${\hat \gamma}_n (y_{n,k})$ in \eqref{eq:lem1}. By applying Bayes' rule to likelihood ratio to obtain the identity $\frac{\mathbb{P}(Y_{n,k} \vert X_{n,k} \neq 0)}{\mathbb{P}(Y_{n,k} \vert X_{n,k} = 0)} = \frac{\mathbb{P}(X_{n,k} \neq  0 \vert Y_{n,k})}{\mathbb{P}(X_{n,k} = 0 \vert Y_{n,k})}\frac{\mathbb{P}(X_{n,k} = 0)}{\mathbb{P}(X_{n,k} \neq 0)} = \frac{{\hat \gamma}_n (Y_{n,k})}{1-{\hat \gamma}_n (Y_{n,k})}\frac{1-\gamma}{\gamma}$, the NP criterion reduces to 
\begin{align}
    g_{\sf NP}(Y_{n,k}) &= \begin{cases}
    0, & \frac{{\hat \gamma}_n (Y_{n,k})}{1-{\hat \gamma}_n (Y_{n,k})} < \frac{\gamma}{1-\gamma}\delta , \\
    1, & \frac{{\hat \gamma}_n (Y_{n,k})}{1-{\hat \gamma}_n (Y_{n,k})} > \frac{\gamma}{1-\gamma}\delta.
    \end{cases} \label{eqn:NP detector 2}
\end{align}
In addition, it is worthwhile mentioning that the MAP detector in \eqref{eqn:MAP detector 2} is equivalent to the NP detector if the threshold parameter is chosen as $\delta = (1-\gamma) / \gamma$.


To find the error probability expressions, we define the region as
\begin{align}
    \mathcal{A}({\sf snr}_n) &= \left\{\, y_{n,k}\in\mathbb{C} \,\middle|\, {\hat \gamma}_n(y_{n,k}) (1-\gamma) > (1-{\hat \gamma}_n(y_{n,k})) \,\gamma \delta \,\right\}.
\end{align}
Then, the false alarm probability $P_{\sf FA}$ and misdetection probability $P_{\sf MD}$ when using the NP detector are given by
\begin{align}
    P_{\sf FA}({\sf snr}_n) &= \int_{y\in\mathcal{A}({\sf snr}_n)} \frac{1}{\pi}e^{-|y|^2} {\rm d}y,\\
    P_{\sf MD}({\sf snr}_n) &= \int_{y\notin\mathcal{A}({\sf snr}_n)} \frac{1}{\pi \left({\sf snr}_n\sigma_{\sf T}^2 + 1\right)}e^{-\frac{\vert y - \sqrt{{\sf snr}_n}\mu_{\sf T}\vert^2}{{\sf snr}_n\sigma_{\sf T}^2 + 1}} {\rm d} y.
\end{align}
To clearly show the generality of the MMSE estimator, we compare the proposed MMSE estimator with the classical NP detector and MAP detector using a receiver operating characteristic (ROC) curve. As shown in Fig.~\ref{fig:ROC}, our MMSE detector with a proper choice threshold can achieve identical ROC curves with the NP detector. This result shows that the MMSE detector is deeply connected to the classical NP detector. At the same time, it helps to define the sensing rate from an information-theoretical viewpoint.

\section{Sensing Rate Characterization}\label{sec:sensing rate}
In this section, we express the sensing rate in terms of estimation-theoretic quantity by applying the I-MMSE relationship in \cite{guo2005mutual}. To provide a better understanding of the sensing rate in terms of target parameters, we present a lower bound of the sensing rate by decomposing it into two parts. Furthermore, we derive approximation of sensing rates with a closed-form. These tractable expressions will be stepping stones toward finding the optimal bandwidth strategy for ISAC systems in the sequel. 

\subsection{An Exact Expression}
The I-MMSE relationship is a foundation that provides a deep connection between an information- and an estimation-theoretic quantity \cite{guo2005mutual}.  The I-MMSE relationship states that the integration of MMSE is the mutual information for an arbitrary target distribution $f_{X_{n,k}}(x_{n,k})$ with a finite moment $\mathbb{E}\left[X_{n,k}^2\right]< \infty$ \cite{guo2005mutual}.  From this statement and the derived MMSE in \eqref{eq:mmse}, we can formally define the sensing rate of a pulse-Doppler radar system.


\begin{theorem} \label{thm2}
The sensing rate $R_{n,k}({\sf snr}_n)=I(X_{n,k};Y_{n,k})$ for $(n,k)$ range-Doppler bin is given by
\begin{align}
    R_{n,k}({\sf snr}_n)     &=  \log_2e \int_0^{{\sf snr}_n}    {\sf mmse}_{n,k} ( s)  {\rm d}s. \label{eq:exact} 
\end{align}

\end{theorem}
\begin{proof}
The proof is the direct consequence from MMSE in Theorem 1 and the I-MMSE relationship in \cite{guo2005mutual}. 
\end{proof}
From Theorem~\ref{thm2}, it is evident that the information rate is an increasing function with ${\sf snr}_n$ by the non-negativity of ${\sf mmse}_{n,k} (s)$ for $s>0$. Although the expression in \eqref{eq:exact} is exact and has full generality, this analytical expression of the sensing rate is intractable to understanding how the target parameters, $(\mu_{\sf T}, \sigma_{\sf T}, \gamma)$, change the information rate.


 \subsection{A Lower Bound}

To shed light on the insight in Theorem~\ref{thm2},  we provide a lower bound of the sensing rate in the following Theorem. 
\begin{theorem} \label{thm3}
Let ${\sf H}_b(p) = -p\log p -(1-p)\log(1-p)$ be a binary entropy function with parameter $p\in [0,1]$. Then, the sensing rate for the $(n,k)$ range-Doppler bin $I(X_{n,k}; Y_{n,k})$ is lower bounded by
\begin{align}\label{eqn:radar lower}
    I(X_{n,k};Y_{n,k}) &\ge {\sf H}_b (\gamma) - {\sf H}_{\sf NP}(\gamma, P_{\sf FA}, P_{\sf MD}) + \gamma\log_2(1+{\sf snr}_n\sigma_{\sf T}^2),
\end{align}
where $P_{\sf FA}$ and $P_{\sf MD}$ is the false alarm and misdetection error probability of \eqref{eqn:NP detector 2} and
\begin{align}
    {\sf H}_{\sf NP}(\gamma, P_{\sf FA}, P_{\sf MD}) = {\sf H}_b(\gamma) + (1-\gamma){\sf H}_b(P_{\sf FA}) + \gamma{\sf H}_b(P_{\sf MD}) - {\sf H}_b\left((1-\gamma)P_{\sf FA} + \gamma(1-P_{\sf MD})\right).
\end{align}

\end{theorem}

\begin{proof}
 We first define a binary random variable $U_{n,k}\in \{0,1\}$ with Bernoulli parameter $\gamma$ to denote the target presence at the $(n,k)$ range-Doppler bin. Similarly, we introduce the target complex amplitude parameter $A_{n,k}$, which is drawn from $\mathcal{CN}(\mu_{\sf T},\sigma_{\sf T}^2).$ Then, our BG target signal $X_{n,k}\sim {\sf BG}(x|\mu_{\sf T},\sigma_{\sf T}^2,\gamma)$ can be represented in a vector form $(U_{n,k},A_{n,k})$. Using the chain rule, we decompose the mutual information as
\begin{align}
    I(X_{n,k} ; Y_{n,k}) &= I\left(U_{n,k}, A_{n,k}; Y_{n,k}\right) \nonumber \\
    &= I(U_{n,k} ; Y_{n,k}) + I\left(A_{n,k} ; Y_{n,k} \vert U_{n,k}\right) \nonumber\\
    &= I(U_{n,k}; Y_{n,k}) + (1-\gamma) I\left(A_{n,k}; Y_{n,k} \vert U_{n,k}=0\right) + \gamma I\left(A_{n,k}; Y_{n,k} \vert U_{n,k}=1\right) \nonumber \\
    &= I(U_{n,k}; Y_{n,k}) + \gamma \log_{2}\left(1+{\sf snr}_n\sigma_{\sf T}^2\right), \label{eqn:sensing rate decomposition}
\end{align}
where the last equality comes from $I(A_{n,k}; Y_{n,k} \vert U_{n,k}=0) = 0$.

We consider the NP detector to obtain the lower bound of $I(U_{n,k}; Y_{n,k})$. Let ${\hat Y}_{n,k}$ be the output of the NP detector. Then, we can form the Markov chain $U_{n,k} \rarr Y_{n,k} \rarr {\hat Y}_{n,k}$. From the processing inequality, we have  $I(U_{n,k};Y_{n,k})\ge I(U_{n,k};{\hat Y}_{n,k})$. By noticing that $U_{n,k}$ and ${\hat Y}_{n,k}$ are the input and output of a binary asymmetric channel, we can lower bound the target detection information as
\begin{align}
    I(U_{n,k}; Y_{n,k}) &\ge I(U_{n,k}; {\hat Y}_{n,k}) \nonumber \\
    &= {\sf H}_b\left((1-\gamma)P_{\sf FA} + \gamma(1-P_{\sf MD})\right) - (1-\gamma){\sf H}_b(P_{\sf FA}) - \gamma{\sf H}_b(P_{\sf MD}).
    \label{eq: target detection lower}
\end{align}
Plugging \eqref{eq: target detection lower} into \eqref{eqn:sensing rate decomposition}, we obtained the expression, which completes the proof.

 
\end{proof}





\begin{figure} 
    \centering
    \subfloat[]{%
    \includegraphics[width=0.33\columnwidth]{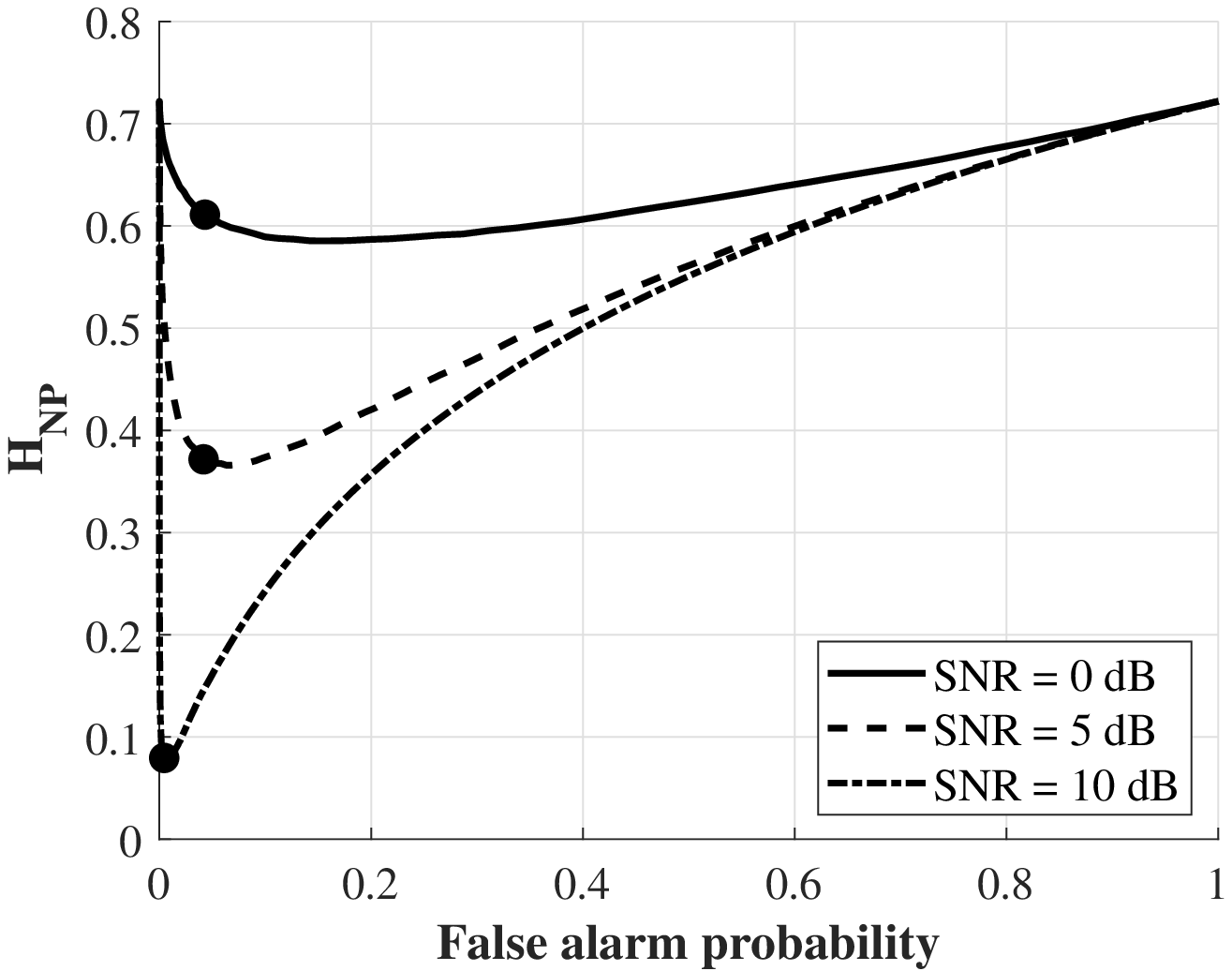}}
    \hfill
    \subfloat[]{%
        \includegraphics[width=0.33\columnwidth]{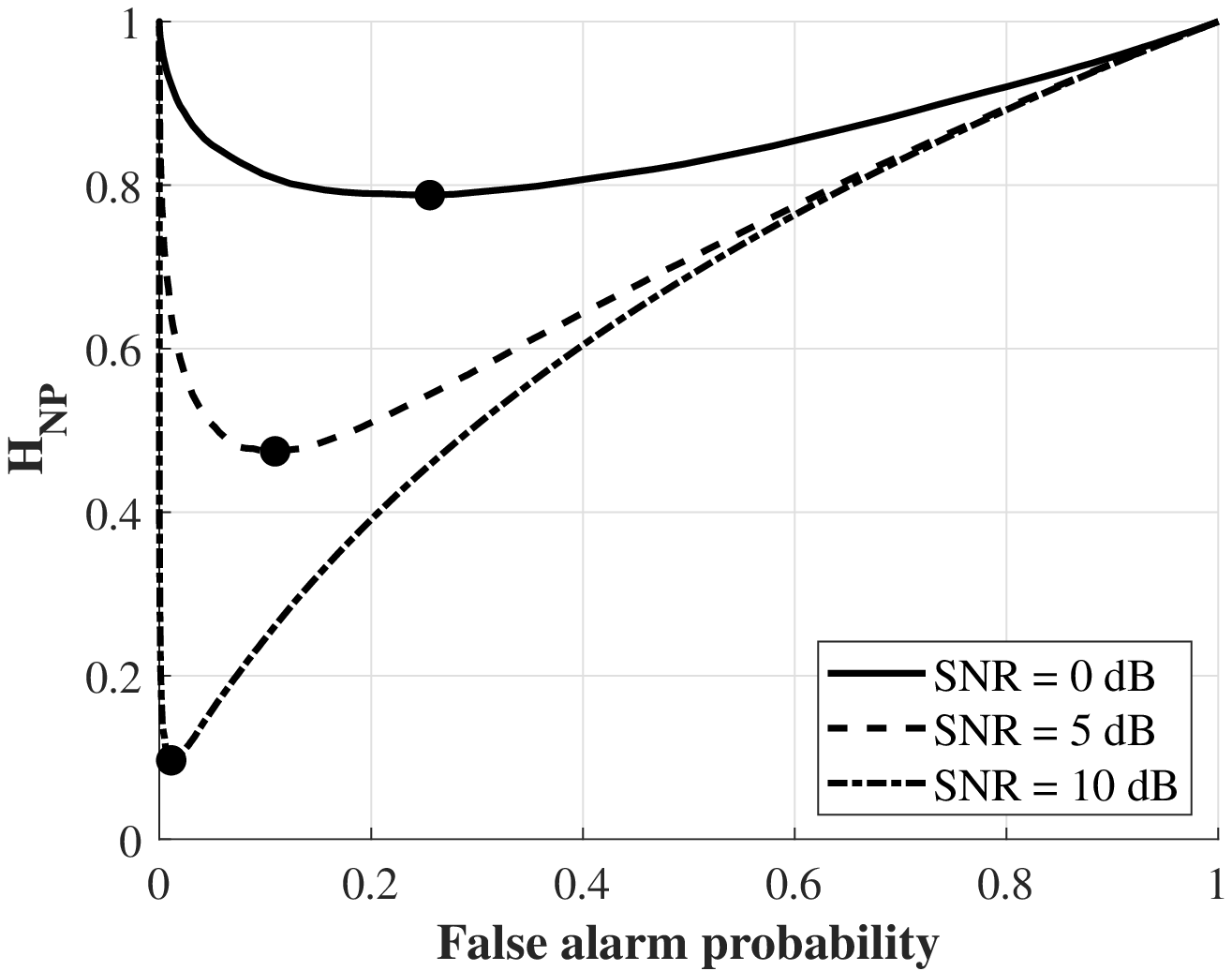}}
    \hfill
    \subfloat[]{%
        \includegraphics[width=0.33\columnwidth]{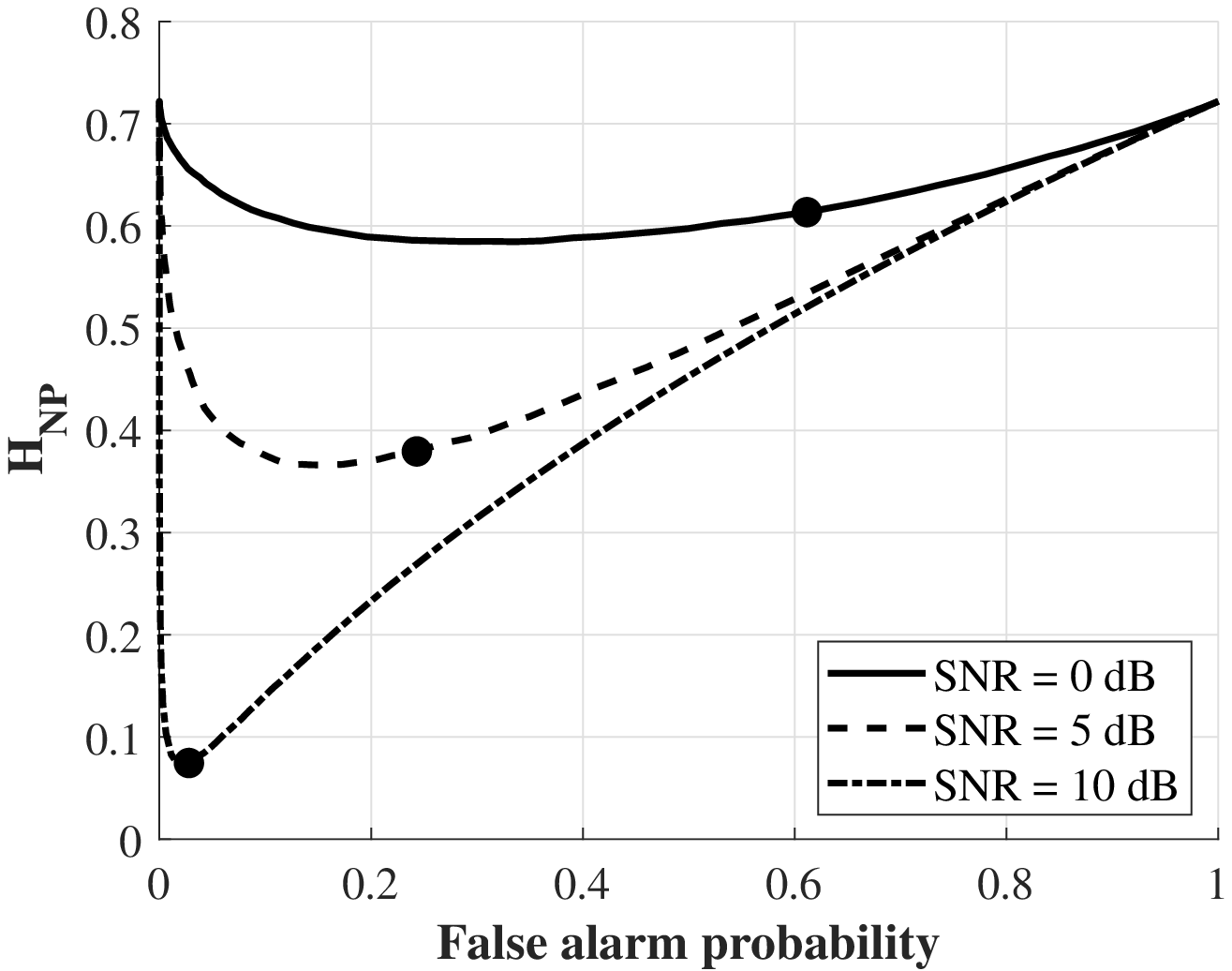}}
    \caption{
    Remaining uncertainty $H_{\sf NP}(\gamma, P_{\sf FA}, P_{\sf MD})$ after detection and false alarm probability $P_{\sf FA}$ when applying NP detector. The circle marker is obtained by MAP detector. Target statistics: $\mu_{\sf T} = 1$, $\sigma_{\sf T}^2 = 0$, ${\sf snr}_n \in \{0, 5, 10\}$ [dB].
    (a) $\gamma = 0.2$,
    (b) $\gamma = 0.5$,
    (c) $\gamma = 0.8$.
    }
\label{fig:detector and entropy}
\end{figure}

Theorem \ref{thm3} elucidates how the sensing rate changes according to the target parameters and target detector. To be specific, the sensing rate can be decomposed into three parts. The first part ${\sf H}_b(\gamma)$ measures the uncertainty of the target localized at a specific ($n,k$) range-Doppler bin. In other words, this part provides the target's location and velocity information in a two-dimensional data matrix. This information rate increases with the underlying target presence probability $\gamma$ in the system, which aligns with our intuition that more target presence likelihood increases the sensing rate. This information rate, however, is reduced by the target detection error, ${\sf H}_{\sf NP}(\gamma, P_{\sf FA}, P_{\sf MD})$. As a result, the sensing rate can diminish as the error probability increases (e.g., low SNR or using sub-optimal detectors) as shown in Fig.~\ref{fig:detector and entropy}.  The last part $ \gamma \log_{2}\left(1+{\sf snr}_n\sigma_{\sf T}^2\right)$ measures the amount of information rate acquired from the target's returned signal under the target presence with probability $\gamma$. This rate is proportional to both the target's presence probability $\gamma$ and the target signal variance $\sigma_{\sf T}^{2}$ for given SNR.

\subsection{Approximation of Sensing Rate}

Although the lower bound in Theorem \ref{thm3} is useful in understanding the sensing rate in terms of the relevant system parameters, it is still unwieldy to compute the sensing rate because the detection error probability computation involves some numerical integrals. 
To remedy this issue,  we provide a tight approximate expression of the sensing rate in closed-form, which is stated in the following theorem. 

\begin{theorem}  \label{thm:sensing rate approximation}
The sensing rate for the $(n,k)$ range-Doppler bin $R_{n,k}({\sf snr}_n)   =I(X_{n,k}; Y_{n,k})$ is tightly approximated as
\begin{align}
    R_{n,k}({\sf snr}_n) &\approx {\sf H}_b(\gamma) - {\hat H}(T \vert Y) + \gamma\log_2\left(1+{\sf snr}_n \sigma_{\sf T}^2\right),
\end{align}
where the approximated conditional entropy ${\hat H}(T\vert Y)$ is given by
\begin{align}
    {\hat H}(T\vert Y) &= (1-\gamma)\log_2\left(1+\frac{\gamma}{1-\gamma}\frac{2}{2+{\sf snr}_n\sigma_{\sf T}^2}e^{-\frac{{\sf snr}_n \vert \mu_{\sf T}\vert^2}{2+{\sf snr}_n\sigma_{\sf T}^2}}\right) \nonumber \\ 
    &\hspace{3em} + \gamma\log_2\left(1+\frac{1-\gamma}{\gamma}\frac{2(1+{\sf snr}_n\sigma_{\sf T}^2)}{2+{\sf snr}_n\sigma_{\sf T}^2}e^{-\frac{{\sf snr}_n \vert \mu_{\sf T}\vert^2}{2 + {\sf snr}_n\sigma_{\sf T}^2}}\right).
\end{align}
\end{theorem}

\begin{proof}
See Appendix~\ref{appendix:proof of sensing rate approximation}.
\end{proof}

As shown in Theorem~\ref{thm:sensing rate approximation}, the approximate rate expression can be composed into three parts as derived in Theorem 3. The difference is the second part, conditional entropy. Although this expression is an approximation of the conditional entropy, it still has full generality to capture the effect of target parameters $(\gamma, \mu_{\sf T}, \sigma_{\sf T}^2)$. In addition, the closed-form expression drastically reduces the computation time required for the evaluation of the sensing rate by facilitating the bandwidth allocation problem for the ISAC system in the next section.

\section{Optimal Bandwidth Allocation for ISAC}\label{sec:ISAC optim}

In this section, we consider the bandwidth allocation problem for the ISAC system. Recall that the weighted sum spectral efficiency in \eqref{eqn:weighted sum rate} involves the trade-off in terms of the bandwidth ratio $\lambda$ because the radar bandwidth $B_{\rm s}$ and the communication bandwidth $B_{\rm c}$ is a function of ratio $\lambda$. In particular, the signal-to-noise ratio of each system (i.e., ${\sf snr}_n$ and ${\sf snr}_c$) relies on the allocated bandwidth $B_{\rm s}$ and $B_{\rm c}$, or ratio $\lambda$. In addition, the number of range bins $N$ is given by the inverse of the radar bandwidth $B_{\rm s}$. For simplicity, we assume that the target parameters $(\gamma, \mu_{\sf T}, \sigma_{\sf T}^2)$ and transmit power $P_{\rm t}$ are constants when allocating the bandwidth, and introduce the following notations to express the dependency on $\lambda$ explicitly:
\begin{align}
    {\sf snr}_{c}(\lambda) &= \frac{\rho_c}{1-\lambda} = \frac{P_{\rm t}}{(1-\lambda)B N_0}, \\
    {\sf snr}_{n}(\lambda) &= \frac{\rho_n}{\lambda} = \frac{M}{ \lambda B N_o}\frac{P_{\rm t}G_{\rm t}\sigma_{\rm rcs} A_{\rm eff}}{\left(4\pi\right)^2 {\bar d}_{n}^4}, \\
    R_{\rm c}(\lambda) &= (1-\lambda) B \log_2\left(1+\frac{\rho_c}{1-\lambda}\right),\\
    R_{\rm s}(\lambda) &= \frac{\lambda B}{N(\lambda)K} \sum_{n=0}^{N(\lambda)-1}\sum_{k=0}^{M-1}R_{n,k}({\sf snr}_n(\lambda) ), \\
    R_{\rm sum}(\lambda) &= \frac{w_{\rm s}R_{\rm s}(\lambda) + w_{\rm c}R_{\rm c}(\lambda)}{B}. \label{eqn:weighted sum rate isac} 
\end{align}

We aim to allocate the bandwidth to maximize the weighted sum spectral efficiency in \eqref{eqn:weighted sum rate isac} by solving the following optimization problem:
\begin{align}
    \lambda^\star &= \operatorname*{argmax}_{0\le \lambda \le 1} R_{\rm sum}(\lambda).
\end{align}
Unfortunately, the objective function $R_{\rm sum}(\lambda)$ is not differentiable for $\lambda$ because the number of addition in $R_{\rm s}(\lambda)$ depends on $\lambda$, which hinders the development of efficient optimization algorithm. 
To circumvent this difficulty and to get an insight into the bandwidth allocation result, we assume that all the range-Doppler bins have the same ${\sf snr}_{s}(\lambda)$. By choosing ${\sf \overline{snr}}_{s}(\lambda) = {\sf snr}_{N-1}(\lambda)$ (i.e., the smallest SNR), we can lower bound the sensing rate in the objective function as
\begin{align}\label{eqn:sensing rate proxy}
    R_{\rm s}(\lambda) &= \frac{\lambda B}{N(\lambda)M}\sum_{n=0}^{N(\lambda)-1}\sum_{k=0}^{M-1} R_{n,k}({\sf snr}_n(\lambda)) \nonumber\\
    &\ge \frac{\lambda B}{N(\lambda)M}\sum_{n=0}^{N(\lambda)-1}\sum_{k=0}^{M-1}R_{n,k}({\sf snr}_{N-1}(\lambda)) \nonumber\\
    &= \lambda B R_{n,k}({\sf \overline{snr}}_{s}(\lambda)) \nonumber\\
    &\overset{\triangle}{=} {\bar R}_{\rm s}(\lambda) 
\end{align}
If we define ${\bar R}_{\rm sum}(\lambda) = w_{\rm s}{\bar R}_{\rm s}(\lambda) + w_{\rm c} R_{\rm c}(\lambda)$, then we have the following bandwidth allocation problem for the ISAC system:
\begin{align}
    {\bar \lambda}^{\star} = \operatorname*{argmax}_{0\le\lambda\le 1} {\bar R}_{\rm sum}(\lambda). \label{eqn:bandwidth allocation simplified}
\end{align}

Now, we present the Karush-Kuhn-Tucker (KKT) optimality conditions for the optimization problem in \eqref{eqn:bandwidth allocation simplified}.
\begin{theorem}\label{thm:bandwidth allocation kkt}
The KKT condition of \eqref{eqn:bandwidth allocation simplified} is satisfied if we choose ${\bar\lambda}^{\star} \in [0, 1]$ such that
\begin{equation}\label{eqn:kkt1}
    w_{\rm s} \left. \frac{d{\bar R}_{\rm s}(\lambda)}{d\lambda}\right \vert_{\lambda = {\bar\lambda}^{\star}}  + w_{\rm c} \left. \frac{dR_{\rm c}(\lambda)}{d\lambda}\right \vert_{\lambda = {\bar\lambda}^{\star}}
    = 0.
\end{equation}
If there is no such ${\bar\lambda}^{\star}\in [0, 1]$, then it is given by 
\begin{align}\label{eqn:kkt2}
    {\bar\lambda}^{\star} = \begin{cases}
    1, &   w_{\rm s} \frac{d{\bar R}_{\rm s}(\lambda)}{d\lambda} > -w_{\rm c}  \frac{dR_{\rm c}(\lambda)}{d\lambda}, \\
    0, &   w_{\rm s} \frac{d{\bar R}_{\rm s}(\lambda)}{d\lambda} < -w_{\rm c}  \frac{dR_{\rm c}(\lambda)}{d\lambda}.
    \end{cases}
\end{align}
In addition, ${\bar R}_{\rm s}(\lambda)$ is increasing concave function with respect to $\lambda$ and $R_{\rm c}(\lambda)$ is monotonic decreasing concave function with respect to $\lambda$.

\end{theorem}
\begin{proof}
See Appendix~\ref{appendix:proof of bandwidth allocation kkt}
\end{proof}


According to Theorem~\ref{thm:bandwidth allocation kkt}, the optimal bandwidth allocation ratio is determined by comparing the derivative of the rates with respect to $\lambda$. For a given ratio $\lambda$, the radar uses more bandwidth ${\rm d}\lambda$ if the increased information of the radar system is larger than the decreased information of the communication system when we set the ratio to be $\lambda + {\rm d}\lambda$. To provide a better understanding of Theorem~\ref{thm:bandwidth allocation kkt}, we offer the estimation theoretic interpretation of the change of the radar information when we allocate infinitesimal additional bandwidth to the radar system. By the I-MMSE relationship, we can express the derivative of the sensing rate as
\begin{align}
    \frac{d R_{n,k}(\overline{\sf snr}_{s}(\lambda))}{d ~ \overline{\sf snr}_{s}(\lambda)} &= {\sf mmse}_{n,k}(\overline{\sf snr}_{s}(\lambda)) ~ \log_{2}{e}.
\end{align}
Using this, we can rewrite the KKT condition as
\begin{align}
    w_{\rm s} \left( R_{n,k}(\overline{\sf snr}_{s}(\lambda)) - \overline{\sf snr}_{s}(\lambda) \cdot {\sf mmse}_{n,k}(\overline{\sf snr}_{s}(\lambda)) \log_{2}{e} \right) = -\frac{w_{\rm c}}{B} \frac{d R_{\rm c}(\lambda)}{d\lambda}. \label{eqn:rewritten kkt}
\end{align}
As a result, the estimation quality of the radar system measured in terms of MMSE is compared with the reduction of the communication information rate when we determine the optimal bandwidth that maximizes the weighted sum information rates. 
In conclusion, it is more efficient for a radar system to use more bandwidth than the communication system in the following two cases: i) the radar estimation error (i.e., MMSE) compared to $\overline{\sf snr}_{s}(\lambda)$ is small enough to compensate the decline of the information rate of the communication system, or ii) the information rate of the radar system is large enough to tolerate a significant estimation error. Finally, we put some remarks about Theorem~\ref{thm:bandwidth allocation kkt}.

\begin{remark}[Bisection method for optimal bandwidth allocation] \rm
Because the left side of \eqref{eqn:kkt1} is decreasing with respect to $\lambda$, the optimal bandwidth allocation ratio $\lambda^{\star}$ can be found by bisection method.
\end{remark}
\begin{remark}[Approximation strategy] \rm
To evaluate the KKT condition in \eqref{eqn:kkt1}, we need to compute complicated integral in \eqref{eq:mmse} and \eqref{eq:exact}, which is time-consuming. Instead, we can use the approximation expression in Theorem~\ref{thm:sensing rate approximation} to compute ${\bar R}_{\rm s}(\lambda)$.
\end{remark}

\section{Numerical results}\label{sec:simulation}
This section presents the numerical simulation results of the sensing rate and ISAC bandwidth allocation problem. 

\subsection{Sensing Rate}
In this subsection, we verify the Theorem~\ref{thm3} and \ref{thm:sensing rate approximation}, each of which characterizes the lower bound and approximation of the sensing rate. In addition, we present the operational meaning of the sensing rate by using the decomposed expression to obtain a more detailed understanding. 

{\bf Tightness of the lower bound and approximation:}
We provide numerical results to verify the exactness of the derived rate expression in Theorem~\ref{thm3} and Theorem~\ref{thm:sensing rate approximation}. Fig.~\ref{fig:approximation tightness} shows that our lower bound and approximation expressions are very tight at all SNRs. This result confirms that our closed-form expression for the sensing rate is exact. In addition, we can use the approximation expression to find the optimal bandwidth allocation strategy for the ISAC system. 

\begin{figure} 
    \centering
    \subfloat[]{%
       \includegraphics[width=0.49\columnwidth]{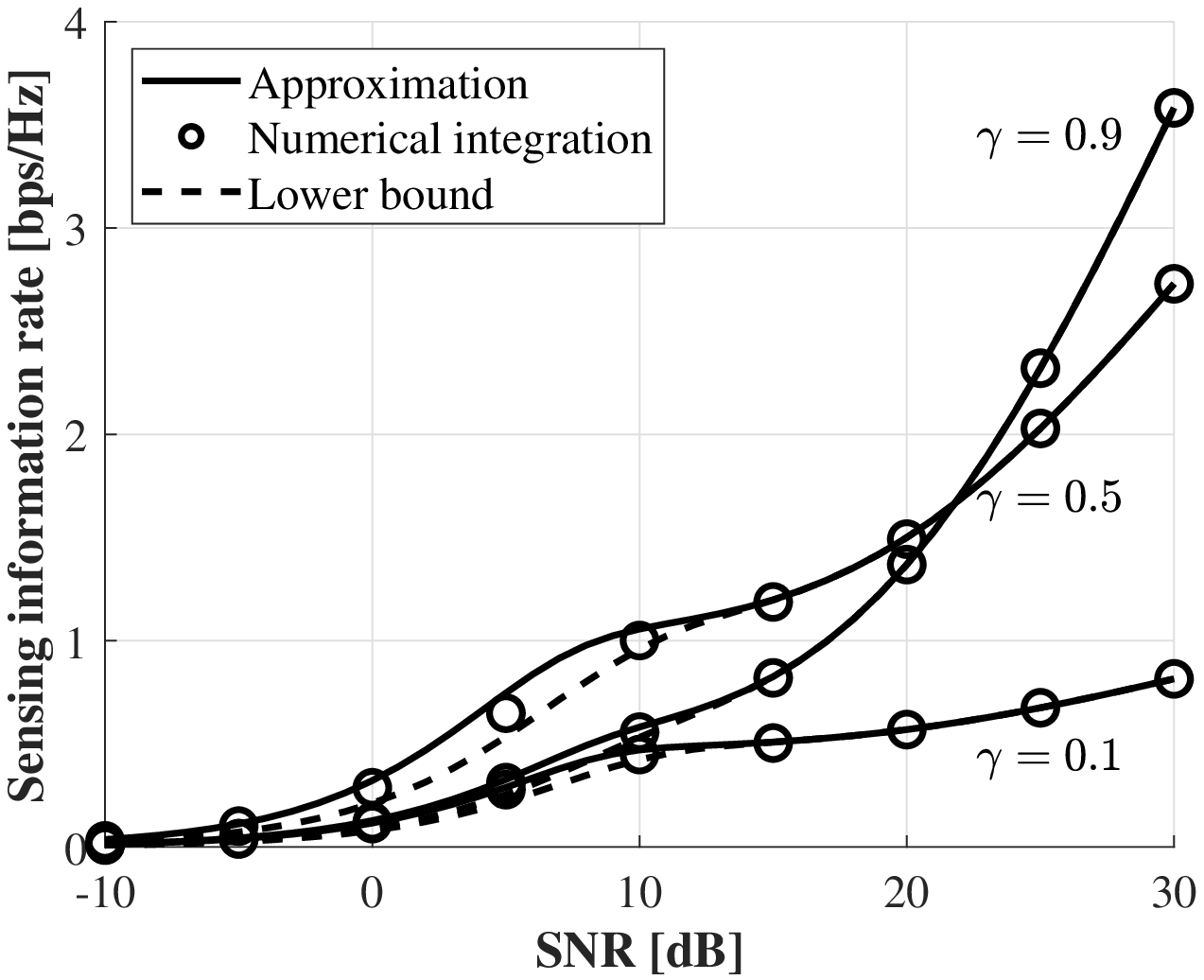}}
    \hfill
    \subfloat[]{%
        \includegraphics[width=0.49\columnwidth]{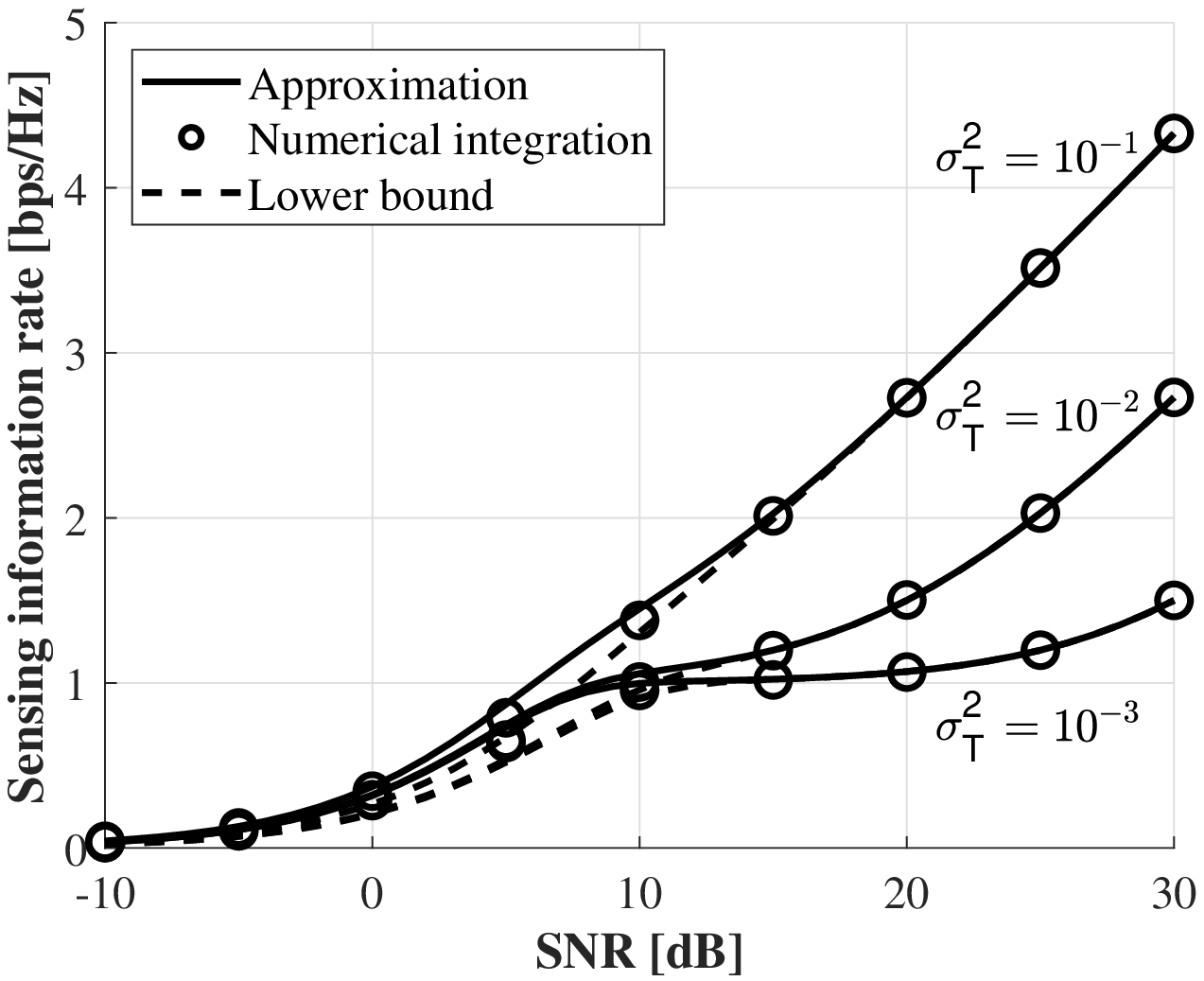}}
    \caption{
    Illustration to verify the tightness of the lower bound and approximation in Theorem~\ref{thm3} and \ref{thm:sensing rate approximation}.
    Target signal statistics: 
    (a) $\gamma \in \{0.1, 0.5, 0.9\}$, $\mu_{\sf T} = 1$, and $\sigma_{\sf T}^2 = 10^{-2}$,
    (b) $\gamma = 0.5$, $\mu_{\sf T} = 1$, and $\sigma_{\sf T}^2 \in \{10^{-1}, 10^{-2}, 10^{-3}\}$.}
    \label{fig:approximation tightness} 
\end{figure}

{\bf The decomposed information:}
To shed light on the behavior of sensing rate, we provide numerical results of {\it target detection information} and {\it target fluctuation information} separately in Fig.~\ref{fig:decomposition}. Recall that the target detection mutual information is $I(U_{n,k}; Y_{n,k})$ and the target fluctuation mutual information is $I(A_{n,k}; Y_{n,k} \vert U_{n,k} = 1)$.

Fig.~{\ref{fig:decomposition}-(a)} presents the decomposed information corresponding to various $\gamma$. We observe that the target detection information is dominant in the low SNR regime, and the target fluctuation information is dominant in the high SNR regime. It means that the target detection ability is important at low SNR. At a high SNR regime, the target signal estimation precision contributes to the sensing rate rather than target detection accuracy. Fig.~{\ref{fig:decomposition}-(b)} is decomposed information displayed by switching $\sigma_{\sf T}^{2}$. The target detection information decreases as the target fluctuation $\sigma_{\sf T}^{2}$ increases because the fluctuation hinders the target detection. Especially for MAP detector in \eqref{eqn:MAP detector 2}, this fluctuation acts as additive noise; thereby, the probability of detection error increases. Since the probability of detection error is related to conditional entropy $H(U_{n,k} \vert Y_{n,k})$, the increased detection error yields the increased conditional entropy, and finally renders the decreased target detection information rate. As the target variance increases, the system can obtain more fluctuation information, and it originates from the increased target signal uncertainty.


\begin{figure} 
    \centering
    \subfloat[]{%
       \includegraphics[width=0.49\columnwidth]{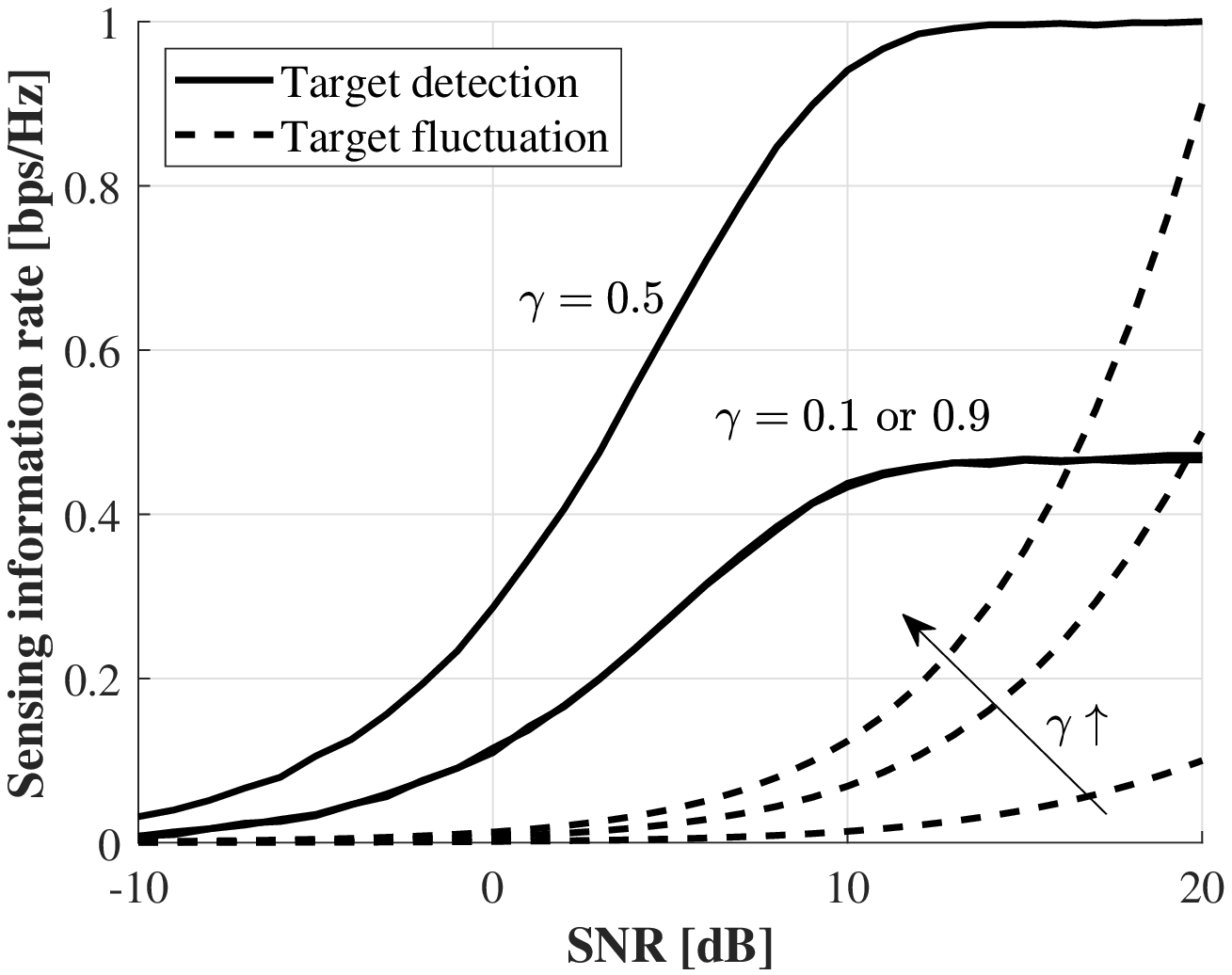}}
    \hfill
    \subfloat[]{%
        \includegraphics[width=0.49\columnwidth]{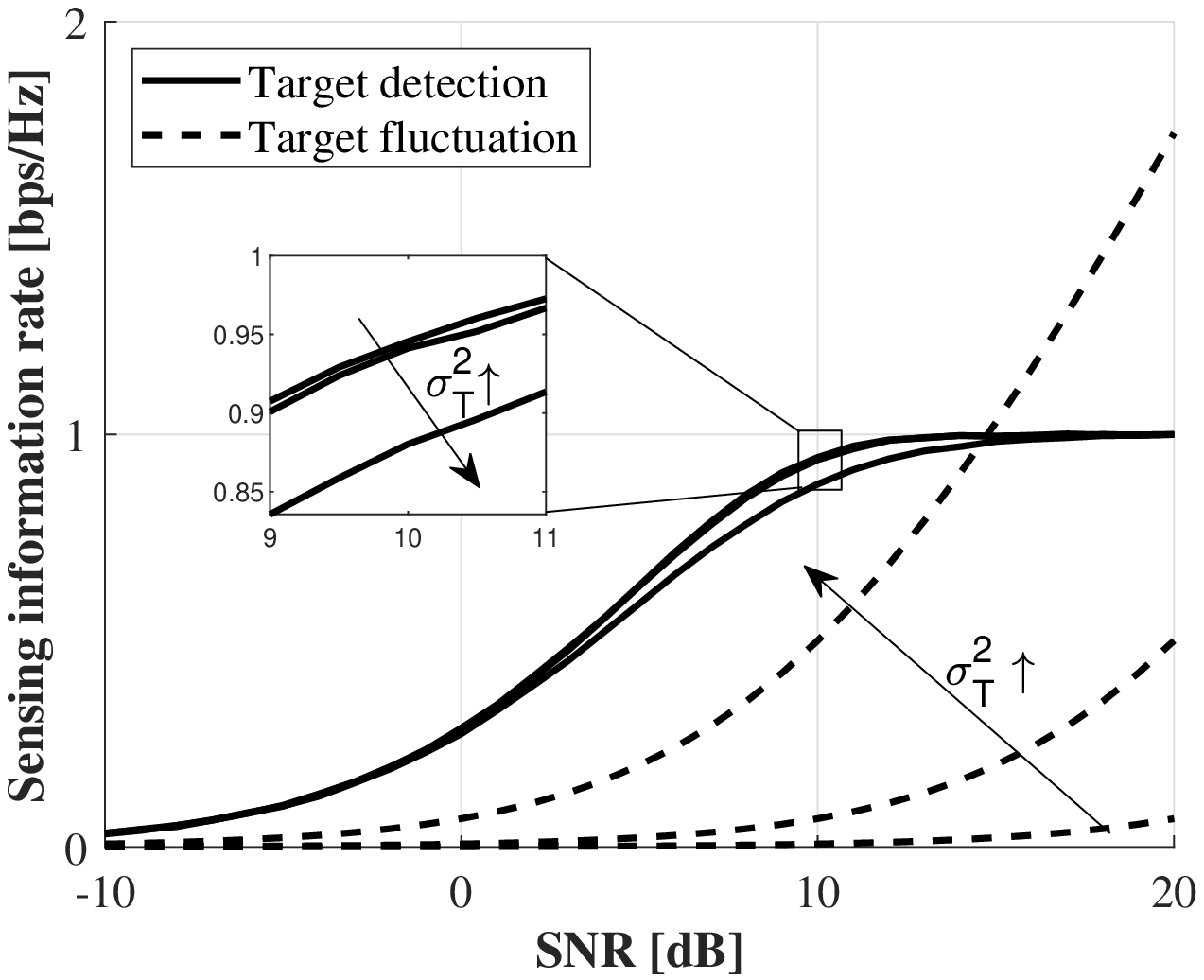}}
    \caption{
    Illustration of decomposed sensing rate. 
    Target signal statistics: 
    (a) $\gamma \in \{0.1, 0.5, 0.9\}$, $\mu_{\sf T} = 1$, and  $\sigma_{\sf T}^2 = 10^{-2}$,
    (b) $\gamma = 0.5$, $\mu_{\sf T} = 1$, and $\sigma_{\sf T}^2 \in \{10^{-1}, 10^{-2}, 10^{-3}\}$.}
    \label{fig:decomposition} 
\end{figure}

\subsection{ISAC Bandwidth Allocation}\label{sec:sim ISAC}
{\bf Simulation setup:}  We assume that a total bandwidth is set to $B=200$ MHz. We set ${\sf snr}_{\rm c}(0) = P_t/(BN_0) = 5$ dB for the communication system. For the radar system, we set ${\sf \overline{snr}_{s}}(1) = {\sf snr}_{N-1}(1) = 15$ dB and the maximum detection range as $10$ km. 

{\bf The effect of target statistics:}
We present the optimal bandwidth ratio $\lambda^\star$ and target statistics ($\gamma$, $\sigma_{\sf T}^{2})$ in Fig.~\ref{fig:KKT optimize}. The upper triangular marker indicates the optimal point. We find that both the amount of weighted sum spectral efficiency ${\bar R}_{\rm sum}(\lambda)$ and the radar-allocated bandwidth increases as both of the target uncertainties $\gamma$ and $\sigma_{\sf T}^{2}$ increase. It is also noticeable that the amount of information increment caused by increased $\sigma_{\sf T}^{2} $ becomes large when $\gamma$ increases due to the larger target estimation information.

\begin{figure} 
    \centering
    {%
    \includegraphics[width=0.5\columnwidth]{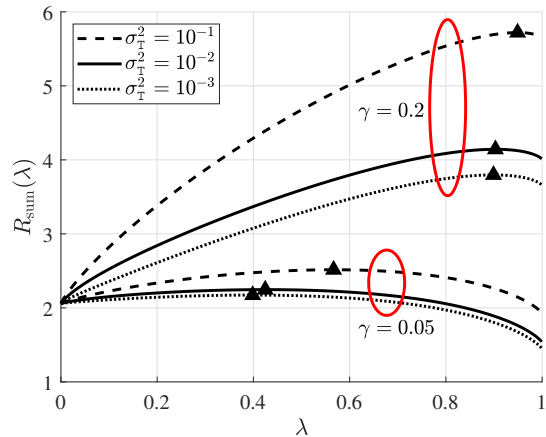}
    }
    \caption{
    Bandwidth allocation results depending on target statistics. The system weights are $w_{\rm s} = 5$ and $w_{\rm c} = 1$.}
    \label{fig:KKT optimize} 
\end{figure}

{\bf ISAC system trade-off:}
Fig.~\ref{fig:trade-off1} presents the information rate trade-off of the ISAC system for various target statistics $\gamma$ and $\sigma_{\sf T}^{2}$. The sensing rate increases gradually when the communication rate decreases. The behavior of the trade-off curve results from how $\gamma$ is modeled depending on the sensing bandwidth and target statistics. In particular, we fix $\gamma$ for each bin. Increasing the sensing bandwidth improves the resolution of the range-Doppler map. The increased resolution gives more chances to discover new targets in the environment. 

\begin{figure} 
    \centering
    \subfloat[]{%
    \includegraphics[width=0.49\columnwidth]{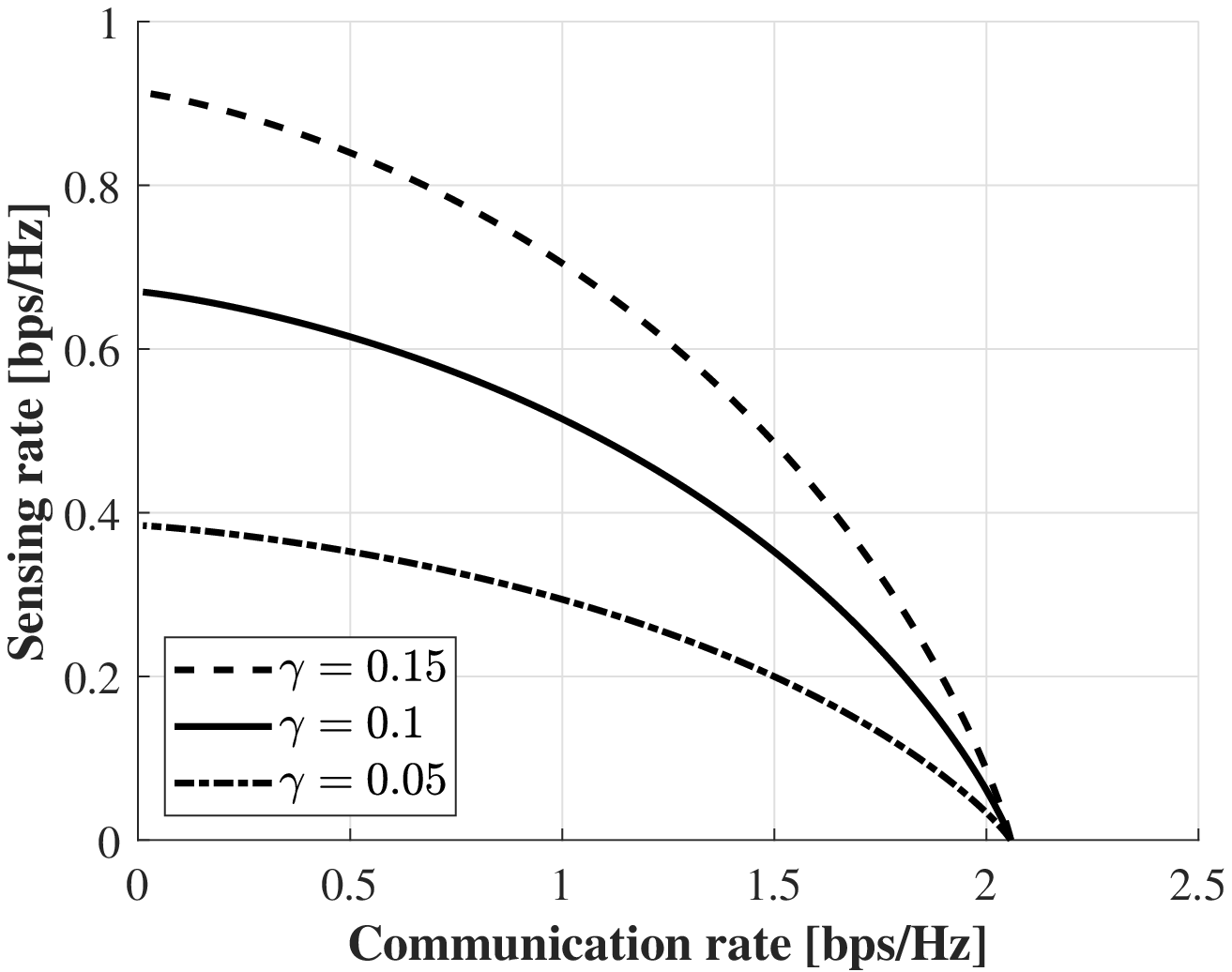}}
    \hfill
    \subfloat[]{%
        \includegraphics[width=0.49\columnwidth]{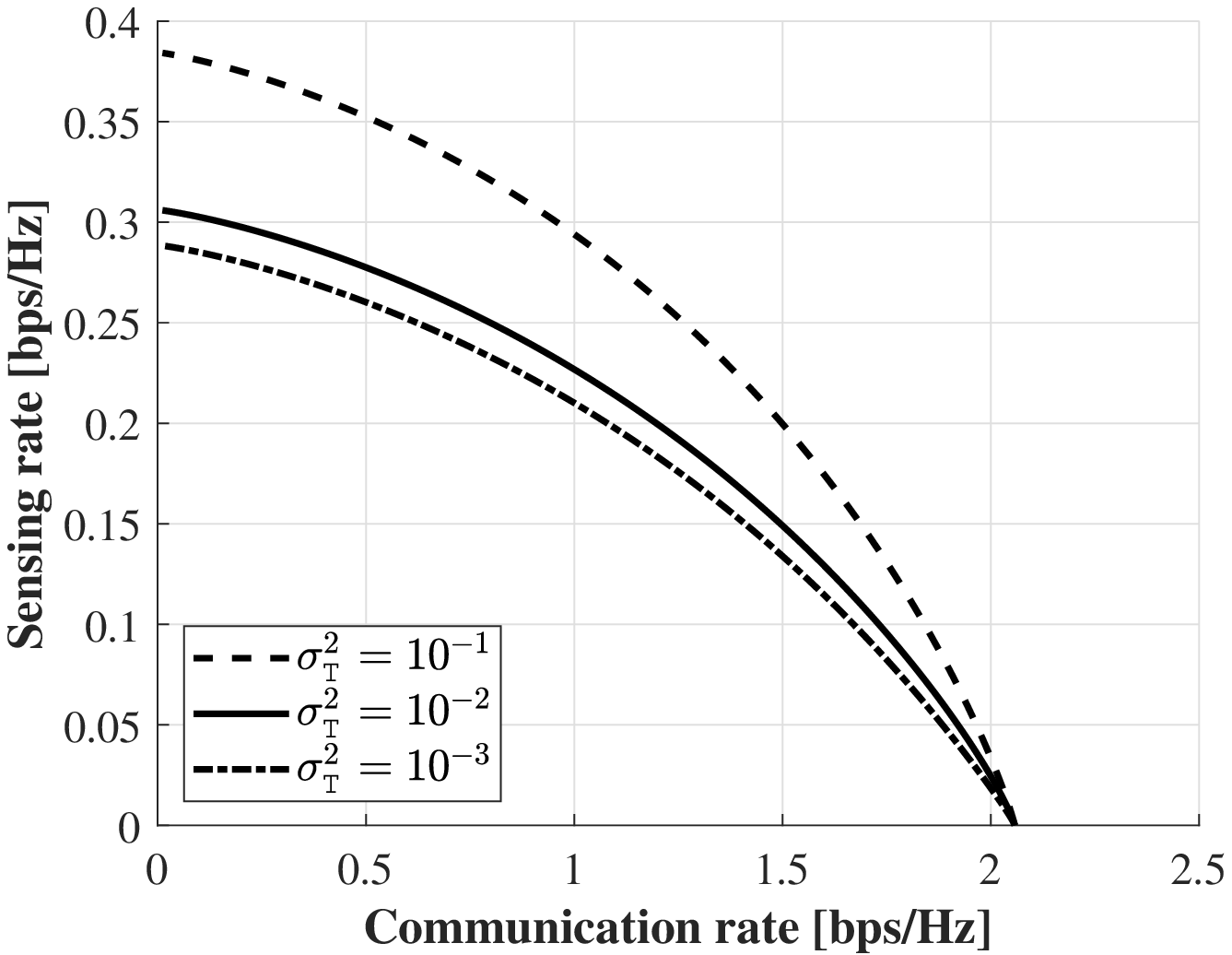}}
    \caption{
    Information rate of ISAC system according to target statistics $\gamma$ and $\sigma_{\sf T}^{2}$. 
    (a) $\gamma = 0.05$,
    (b) $\sigma_{\sf T}^{2} = 10^{-1}$.
    }
    \label{fig:trade-off1} 
\end{figure}

\section{Conclusion}\label{conclusion}
By utilizing an information-theoretic approach, we evaluated the performance of the ISAC system by characterizing the amount of information that the pulse-Doppler radar can obtain. First, we calculated the sensing rate of an unknown environment using the pulse-Doppler radar and established correlations with various radar system parameters. Then, by integrating the sensing rate, closed-form expressions, and Shannon's communication rate into the ISAC system, we derived an optimal bandwidth allocation strategy that maximizes the weighted sum spectral efficiency. Our findings indicated that when target uncertainty and SNR are high, more bandwidth should be allocated for the radar system. Our simulation results confirmed the accuracy of the derived closed-form expressions.

\appendix
\subsection{Proof of Lemma~\ref{lem1}}\label{appendix:proof of mmse estimator}
From Baye's rule, we compute 
\begin{align}
	\mathbb{P}[X_{n,k}\neq 0 |Y_{n,k}=y_{n,k}]  
	&=\frac{\mathbb{P}[X_{n,k}\neq 0|Y_{n,k}=y_{n,k}]\mathbb{P}(Y_{n,k}=y_{n,k})}{\mathbb{P}(Y_{n,k}=y_{n,k})}\nonumber\\
	&=\frac{ \mathbb{P}(X_{n,k}\neq 0, Y_{n,k}=y_{n,k})}{\mathbb{P}(X_{n,k}=0,Y_{n,k}=y_{n,k})+\mathbb{P}(X_{n,k}\neq 0,Y_{n,k}=y_{n,k})}. \label{eq:inclusionprob}
\end{align}
The joint probabilities are computed as
\begin{align}
 \mathbb{P}(X_{n,k}=0,Y_{n,k}=y_{n,k}) &=\mathbb{P}(X_{n,k}=0)\mathbb{P}(Y_{n,k}=y_{n,k}|X_{n,k}=0) \nonumber\\
 &=(1-\gamma) \frac{1}{\pi}e^{-\vert y_{n,k}\vert^2} \label{eq:joint1}
\end{align}
and
\begin{align}
 \mathbb{P}(X_{n,k}\neq 0,Y_{n,k}=y_{n,k}) 	&=\mathbb{P}(X_{n,k}\neq 0)\mathbb{P}(Y_{n,k}=y_{n,k}|X_{n,k}\neq 0) \nonumber\\
 &= \gamma \times \frac{1}{\pi ({\sf snr}_n\sigma_{\sf T}^2+1)}e^{-\frac{\vert y_{n,k}-\sqrt{{\sf snr}_n}\mu_{\sf T}\vert^2}{{\sf snr}_n\sigma_{\sf T}^2+1}}. \label{eq:joint2}
\end{align}
Invoking \eqref{eq:joint1} and \eqref{eq:joint2} into \eqref{eq:inclusionprob}, we obtain the probability that the target is absent in the $(n,k)$ range-Doppler bin as
\begin{align}
		\mathbb{P}[X_{n,k}\neq 0 |Y_{n,k}=y_{n,k}] 
		& = \frac{ \gamma \frac{1}{\pi ({\sf snr}_n \sigma_{\sf T}^2+1)}e^{-\frac{\vert y_{n,k}-\sqrt{{\sf snr}_n}\mu_{\sf T}\vert^2}{{\sf snr}_n\sigma_{\sf T}^2+1}}} { (1-\gamma) \frac{1}{\pi}e^{-\vert y_{n,k}\vert^2}
+ \gamma \frac{1}{\pi ({\sf snr}_n\sigma_{\sf T}^2+1)}e^{-\frac{\vert y_{n,k}-\sqrt{{\sf snr}_n}\mu_{\sf T}\vert^2}{{\sf snr}_n\sigma_{\sf T}^2+1}}}&\nonumber\\
&=\frac{1}{1+\frac{1-\gamma}{\gamma}({\sf snr}_n\sigma_{\sf T}^2+1)e^{-\vert y_{n,k}\vert^2+\frac{\vert y_{n,k}-\sqrt{{\sf snr}_n}\mu_{\sf T}\vert^2}{{\sf snr}_n\sigma_{\sf T}^2+1}} }. \label{eq:inclusion_pro}
\end{align}
This completes the proof of lemma.

\subsection{Proof of Theorem~\ref{thm:sensing rate approximation}}\label{appendix:proof of sensing rate approximation}
The proof proceeds two-steps: i) finding a lower bound of entropy of $Y_{n,k}$ and ii) compensating an offset that occurred in an asymptotic regime, i.e., ${\sf snr}_n \rarr 0$ and ${\sf snr}_n \rarr \infty$. In this proof, we shall use the equivalent channel:
\begin{align}
    Y_{n,k} &= X_{n,k} + Z_{n,k}, \qquad Z_{n,k} \sim \mathcal{CN}\left(0, \frac{1}{{\sf snr}_n}\right).
\end{align}
In addition, we drop the indices $(n,k)$ for notational convenience. 

First, the distribution of $Y$ is given by 
\begin{align}
    p(y) &= (1-\gamma)\frac{{\sf snr}}{\pi}\exp\left(-\rho\vert y\vert^2\right) + \gamma\frac{1}{\pi\left(\sigma_{\sf T}^2 + \frac{1}{{\sf snr}}\right)}\exp\left(-\frac{\vert y - \mu_{\sf T}\vert^2}{\sigma_{\sf T}^2 + \frac{1}{{\sf snr}}}\right)  \nonumber \\
    &\overset{\triangle}{=} \sum_{i=0}^1 w_ip_i(y),\label{eqn:mixture Gaussian pdf}
\end{align}
where $w_0=1-\gamma$, $w_1=\gamma$, $p_0(y)\sim\mathcal{CN}(0, 1/{\sf snr})$, and $p_1(y)\sim\mathcal{CN}(\mu_{\sf T}, \sigma_{\sf T}^2 + (1/{\sf snr}))$. 

From \cite{kolchinsky2017estimating}, a lower bound of $H(Y)$ is computed as 
\begin{align}\label{eqn:applying jensen}
    H(Y) 
    &= -\sum_i w_i \int p_i(y) \log \sum_j w_j p_j(y) {\rm d}y \nonumber \\
    &\ge -\sum_i w_i \log \int p_i(y) \sum_j w_j p_j(y) {\rm d}y \nonumber \\
    &= -\sum_i w_i \log \sum_j w_j \underbrace{\left( \int p_i(y)p_j(y) {\rm d}y \right)}_{\epsilon_{i,j}} \nonumber \\
    &\overset{\triangle}{=} H_{\sf Low}(Y).
\end{align}

For the complex Gaussian distribution $p_i(y)$ and $p_j(y)$, the integral $\epsilon_{i,j}$ in \eqref{eqn:applying jensen} can be further evaluated as follows: 
\begin{align}\label{eqn:P}
    \epsilon_{i,j}
    &= \frac{1}{\pi^2 \sigma_i^2 \sigma_j^2}\int_{y\in\mathbb{C}} \exp\left(-\frac{\vert y-\mu_i\vert^2}{\sigma_i^2}-\frac{\vert y-\mu_j\vert^2}{\sigma_j^2}\right) {\rm d}y \nonumber \\
    &= \frac{1}{\pi^2 \sigma_i^2 \sigma_j^2}\exp\left({-\frac{\sigma_j^2 \vert \mu_i\vert^2 + \sigma_i^2 \vert \mu_j\vert^2}{\sigma_i^2\sigma_j^2}}\right)  \nonumber \\
    & \qquad \times \underbrace{\int_{y\in\mathbb{C}} 
    \exp\left[
    -\left(\frac{1}{\sigma_i^2}+\frac{1}{\sigma_j^2}\right)\vert y \vert^2 + \frac{2}{\sigma_i^2}\text{Re}\{\mu_i^* y\} + \frac{2}{\sigma_j^2}\text{Re}\{\mu_j^* y\}
    \right] {\rm d}y}_{Q},
\end{align}
where $\mu_i^*$ and $\mu_j^*$ are complex conjugate of $\mu_i$ and $\mu_j$ respectively.
For the complex number $y$, by denoting the real component as $y_R$, and imaginary component $y_I$, the integration of $Q$ is separated into real and imaginary parts
\begin{align}
    Q &= \int_{\infty}^{\infty} \exp\left(-\left(\frac{1}{\sigma_i^2}+\frac{1}{\sigma_j^2}\right)y_R^2 + 2\left(\frac{\mu_{i,R}}{\sigma_i^2}+\frac{\mu_{j,R}}{\sigma_j^2}\right) y_R\right) {\rm d}y_R \nonumber \\
    & \qquad \times \int_{\infty}^{\infty} \exp\left(-\left(\frac{1}{\sigma_i^2}+\frac{1}{\sigma_j^2}\right)y_I^2 + 2\left(\frac{\mu_{i,I}}{\sigma_i^2}+\frac{\mu_{j,I}}{\sigma_j^2}\right) y_I\right) {\rm d}y_I. \label{eqn:Q} 
\end{align}

The integral \eqref{eqn:Q} are computed using the following identity:
\begin{equation}\label{eqn:Q identity}
        \int_{-\infty}^{\infty} \exp\left(-ay^2+2by\right) {\rm d}y = \sqrt{\frac{\pi}{a}}\exp\left(\frac{b^2}{a}\right), \ \ a > 0.
\end{equation}

Then, $Q$ is computed as
\begin{align}
    Q 
    &= \frac{\pi \sigma_i^2\sigma_j^2}{\sigma_i^2 + \sigma_j^2}\exp\left(
    \frac{\vert \mu_i\vert^2\sigma_j^4 + \vert\mu_j\vert^2\sigma_i^4 + 2\sigma_i^2\sigma_j^2(\mu_{i,R}\mu_{j,R}+\mu_{i,I}\mu_{j,I})}{\sigma_i^2\sigma_j^2(\sigma_i^2 +\sigma_j^2)}
    \right). \label{eqn:Q2} 
\end{align}

Using \eqref{eqn:P}, \eqref{eqn:Q identity}, and \eqref{eqn:Q2} we have
\begin{align}
    \epsilon_{i,j} 
    &= \frac{1}{\pi(\sigma_i^2+\sigma_j^2)}\exp\left(
    -\frac{\vert \mu_i - \mu_j\vert^2}{\sigma_i^2+\sigma_j^2}
    \right).
\end{align}
Therefore, we obtain the lower bound of $H(Y)$ as 
\begin{equation}\label{eqn:entropy lb tmp}
    H(Y)\ge H_{\sf Low}(Y) = -\sum_i w_i \log \sum_j w_j \epsilon_{i,j}.
\end{equation}

Next, we compensate for the offset in asymptotic regimes as done in \cite{choi2018spatial}. For this purpose, we limit ourselves to the case of $\gamma=0.5$ and $\sigma_{\sf T}^2 = 0$. Applying \eqref{eqn:entropy lb tmp} gives a lower bound of mutual information as
\begin{align}
    I(X; Y) &\ge H_{\sf Low}(Y) - H(Y\vert X) \nonumber \\
    &= -\log\left(\frac{{\sf snr}+ {\sf snr} \exp\left(-\frac{{\sf snr}\vert \mu_{\sf T}\vert^2}{2}\right)}{4\pi}\right) - \log(\pi e /{\sf snr}) \nonumber \\
    &= -\log\left(\frac{e}{4} \left(1+\exp\left(-\frac{{\sf snr}\vert \mu_{\sf T}\vert^2}{2}\right)\right) \right) \nonumber \\
    &\overset{\triangle}{=} I_{\sf Low}(X;Y).
\end{align}
Taking the limits to the two extreme SNR regimes gives the limiting values:
\begin{align}\label{eqn:limiting value}
    \lim_{\rho \rarr 0} I_{\sf Low}(X;Y) &= \log(2/e), \nonumber \\
    \lim_{\rho \rarr \infty} I_{\sf Low}(X;Y) &= \log(4/e).
\end{align}
Notice that the distribution of $Y$ is statistically equivalent to the received signal when sending a BPSK signal over an AWGN channel. Then, intuitively, the mutual information $I(X;Y)$ should approach to 0 when ${\sf snr}=0$, and to $\log 2$ when ${\sf snr}=\infty$. Using these facts and two limiting values in \eqref{eqn:limiting value}, we can define the offset of approximation as
\begin{equation}\label{eqn:lower bound offset}
    \Delta_{I(X;Y)} \overset{\triangle}{=} I(X;Y) - I_{\sf Low}(X;Y) = -\log(2/e). 
\end{equation}
Adding the offset \eqref{eqn:lower bound offset} to \eqref{eqn:entropy lb tmp} yields a tight approximation as
\begin{align}
    I(X;Y) &\approx H_{\sf Low}(Y) + \Delta_{I(X;Y)} - H(Y\vert X) \nonumber \\
    &= -\sum_i w_i \log\sum_j w_j \epsilon_{i,j} - \log(2\pi/{\sf snr}). \label{eqn:rate approx dirty}
\end{align}

Recall that the coefficients are $w_0=1-\gamma$, $w_1=\gamma$, $p_0(y)\sim\mathcal{CN}(0, 1/{\sf snr})$, and $p_1(y)\sim\mathcal{CN}(\mu_{\sf T}, \sigma_{\sf T}^2 + (1/{\sf snr}))$. By substituting these coefficients into \eqref{eqn:rate approx dirty}, we have the following:
\begin{align}
    I(X;Y) &\approx -(1-\gamma)\log\left( \frac{{\sf snr}~(1-\gamma)}{2\pi}  + \frac{{\sf snr}~\gamma}{\pi(2+{\sf snr}\sigma_{\sf T}^2)} e^{-\frac{{\sf snr}\vert\mu_{\sf T}\vert^2}{2+{\sf snr}\sigma_{\sf T}^2}} \right) - \nonumber \\
    & \hspace{2em} -\gamma \log\left( \frac{{\sf snr}~\gamma}{2\pi(1+{\sf snr}\sigma_{\sf T}^2)} + \frac{{\sf snr}~(1-\gamma)}{\pi (2+{\sf snr}\sigma_{\sf T}^2)} e^{-\frac{{\sf snr}\vert\mu_{\sf T}\vert^2}{2+{\sf snr}\sigma_{\sf T}^2}} \right) - \log\left(\frac{2\pi}{{\sf snr}}\right) \nonumber \\
    &= \gamma\log(1+{\sf snr}\sigma_{\sf T}^2) + {\sf H}_b(\gamma) - (1-\gamma) \log\left(1 + \frac{2\gamma}{(1-\gamma)(2+{\sf snr}\sigma_{\sf T}^2)}e^{-\frac{{\sf snr}\vert\mu_{\sf T}\vert^2}{2+{\sf snr}\sigma_{\sf T}^2}} \right) \nonumber \\
    &\hspace{2em} -\gamma\log\left(
    1+\frac{2(1-\gamma)(1+{\sf snr}\sigma_{\sf T}^2)}{\gamma(2+{\sf snr}\sigma_{\sf T}^2)}e^{-\frac{{\sf snr}\vert\mu_{\sf T}\vert^2}{2+{\sf snr}\sigma_{\sf T}^2}} \right).
\end{align}
This completes the proof.

\subsection{Proof of Theorem~\ref{thm:bandwidth allocation kkt}}\label{appendix:proof of bandwidth allocation kkt}
    Let $f_1(\lambda) = w_{\rm s}{\bar R}_{\rm s}(\lambda)$ and $f_2(\lambda) = w_{\rm c}{R}_{\rm c}(\lambda)$. 
    First, we will show that both $f_1(\lambda)$ and $f_2(\lambda)$ are concave in $\lambda \in [0, 1]$. Then, satisfying \eqref{eqn:kkt1} is sufficient to be the global optimal because ${\bar R}_{\rm sum}(\lambda) = f_1(\lambda) + f_2(\lambda)$ is also concave. When there is no such ${\bar\lambda}^\star$, the selection rule according to \eqref{eqn:kkt2} is trivially obtained because i) $\frac{d {\bar R}_{\rm sum}(\lambda)}{d\lambda} > 0$ implies that increasing $\lambda$ from 0 to 1 increases ${\bar R}_{\rm sum}(\lambda)$ and ii) $\frac{d {\bar R}_{\rm sum}(\lambda)}{d\lambda} < 0$ implies that decreasing $\lambda$ from 1 to 0 increases ${\bar R}_{\rm sum}(\lambda)$. 
    Therefore, the remaining task is to prove that $f_1(\lambda)$ and $f_2(\lambda)$ are concave in $\lambda \in [0, 1]$ by showing that the second order derivative is negative.

    First, we consider $f_1(\lambda)$. 
    Recall that ${\bar R}_{\rm s}(\lambda) = \lambda B R_{n,k}( {\sf \overline{snr}_{s}(\lambda)} )$ and ${\sf \overline{snr}_{s}}(\lambda) = \rho_{N-1}/\lambda$. We use the following I-MMSE relationship to derive the derivative expression:
    \begin{align}
        \frac{d R_{n,k}(\overline{\sf snr}_{s}(\lambda))}{d ~ \overline{\sf snr}_{s}(\lambda)} &= {\sf mmse}_{n,k}(\overline{\sf snr}_{s}(\lambda)) ~ \log_{2}{e}.
    \end{align}
    Then, the first-order derivative of $f_1(\lambda)$ is given by 
    \begin{align}
        \frac{df_1(\lambda)}{d\lambda} &= w_{\rm s}B\left(R_{n,k}(\overline{\sf snr}_{s}(\lambda)) -\overline{\sf snr}_{s}(\lambda) \cdot \frac{dR_{n,k}(\overline{\sf snr}_{s}(\lambda))}{d~\overline{\sf snr}_{s}(\lambda)}\right) \nonumber \\
        &= w_{\rm s}B\left(R_{n,k}(\overline{\sf snr}_{s}(\lambda)) - \overline{\sf snr}_{s}(\lambda) \cdot {\sf mmse}_{n,k}(\overline{\sf snr}_{s}(\lambda))\log_2e\right) \ge 0, 
    \end{align}    
    where ${\sf mmse}_{n,k}$ is in \eqref{eq:mmse}.
    Also, the non-negativeness come from the I-MMSE relationship in \cite[Corollary 2]{guo2005mutual}.
    The second-order derivative is given by
    \begin{equation}\label{eqn:second deriv 1}
        \frac{d^2 f_1(\lambda)}{d\lambda^2} = \frac{w_{\rm s}\left(\overline{\sf snr}_{s}(\lambda)\right)^2}{\lambda} \cdot \frac{d~{\sf mmse}_{n,k}({\tilde \rho}(\lambda))}{d~\overline{\sf snr}_{s}(\lambda)}\log_2e < 0.
    \end{equation}
    Here, the inequality comes from the fact that MMSE is decreasing function of $\overline{\sf snr}_{r}(\lambda)$ \cite{guo2011estimation}. 
    
    
    Second, we consider $f_2(\lambda)$. The direct computation yields the following first and second order derivative:
    \begin{align}
        \frac{d f_2(\lambda)}{d\lambda} &= -w_{\rm c}B \log_2\left(1+\frac{{\rho}_{\rm c}}{1-\lambda}\right) + \frac{w_{\rm c}B{\rho}_{\rm c}}{1-\lambda + {\rho}_{\rm c}}\log_2e < 0, \nonumber \\
        \frac{d^2 f_2(\lambda)}{d\lambda^2} &= - \frac{w_{\rm c}B{\rho}_{\rm c}^2}{(1-\lambda)\left(1-\lambda+{\rho}_{\rm c}\right)^2}\log_2e < 0. \label{eqn:second deriv 2}
    \end{align}
    Because of \eqref{eqn:second deriv 1} and \eqref{eqn:second deriv 2}, both $f_1(\lambda)$ and $f_2(\lambda)$ are concave in $\lambda \in [0,1]$.

\bibliographystyle{IEEEtran}
\bibliography{refs_all}

\end{document}